\documentclass[a4paper,11pt,UKenglish]{article}

\usepackage{bm,mathptmx,booktabs}
\usepackage{tcolorbox}
\usepackage[usestackEOL]{stackengine}

\usepackage{thm-restate}
\usepackage{amsmath,amssymb}

\DeclareMathAlphabet{\mathcal}{OMS}{cmsy}{m}{n}

\usepackage{tikz}
\usetikzlibrary{decorations.pathreplacing}

\usepackage{ifthen}
\usepackage{comment}
\makeatletter
\usepackage[sort,comma,square,numbers]{natbib}
\renewcommand*{\NAT@spacechar}{~} %
\renewcommand\bibsection %
{
  \section*{\refname
    \@mkboth{\MakeUppercase{\refname}}{\MakeUppercase{\refname}}}
}

\usepackage{color,tikz,environ}
\usetikzlibrary{decorations.markings}
\usepackage{etoolbox}

\usepackage{titlesec}

\titleformat{\paragraph}
{\normalfont\normalsize\bfseries}{\theparagraph}{1em}{}
\titlespacing*{\paragraph}
{0pt}{3.25ex plus 1ex minus .2ex}{1.5ex plus .2ex}


\newcommand{\boundellipse}[3]%
{(#1) ellipse (#2 and #3)
}

\usepackage[margin=1in]{geometry}

\ifthenelse{\isundefined{\lipics}}
{
\usepackage{graphicx}
\usepackage[colorlinks=true,bookmarks=false]{hyperref}
\usepackage[small,bf]{caption}
\usepackage{subfigure}
\usepackage[british]{babel}
}
{}
\usepackage{xcolor}
\definecolor{darkblue}{rgb}{0,0,0.45}
\definecolor{darkred}{rgb}{0.6,0,0}
\definecolor{darkgreen}{rgb}{0.13,0.5,0}
\hypersetup{colorlinks, linkcolor=darkblue, citecolor=darkgreen,
urlcolor=darkblue}

\ifthenelse{\isundefined{\llncs}}{
  \usepackage{amsthm}
  \usepackage{authblk}
  
}

\usepackage{amsmath,amsfonts,amssymb}
\usepackage{mathtools}
\usepackage{mathrsfs}
\usepackage{wrapfig}
\usepackage{xspace}
\usepackage[shortlabels]{enumitem}
\setlist[enumerate]{nosep} %
\setlist[itemize]{nosep} %

\usepackage{multirow}

\usepackage[protrusion=true]{microtype}

\usepackage{aliascnt}
\ifthenelse{\isundefined{\llncs}}{
  \theoremstyle{plain}
}{}

\newaliascnt{lemma}{theorem}
\newaliascnt{corollary}{theorem}
\newaliascnt{definition}{theorem}
\newaliascnt{claim}{theorem}
\newaliascnt{proposition}{theorem}
\newaliascnt{remark}{theorem}
\newaliascnt{hypothesis}{theorem}
\newaliascnt{observation}{theorem}
\newtheorem{lemma}[lemma]{Lemma}
\newtheorem{claim}[claim]{Claim}

\newtheorem{remark}[remark]{Remark}

\ifthenelse{\isundefined{\llncs}}{
  \theoremstyle{definition}
}{}
\newtheorem{definition}[definition]{Definition}
\aliascntresetthe{lemma}
\aliascntresetthe{remark}
\aliascntresetthe{corollary}
\aliascntresetthe{proposition}
\aliascntresetthe{definition}
\aliascntresetthe{claim}
\aliascntresetthe{hypothesis}
\aliascntresetthe{observation}

\newcommand{\ignore}[1]{}
\usepackage{environ,xstring}
\newif\iflabel
\newif\ifdbs
\newif\ifamp
\NewEnviron{doitall}{%
  \noexpandarg
  \expandafter\IfSubStr\expandafter{\BODY}{\label}{\labeltrue}{\labelfalse}%
  \expandafter\IfSubStr\expandafter{\BODY}{\\}{\dbstrue}{\dbsfalse}%
  \expandafter\IfSubStr\expandafter{\BODY}{&}{\amptrue}{\ampfalse}%
  \iflabel\def\doitallstar{}\else\def\doitallstar{*}\fi
  \ifdbs
    \ifamp
      \def\doitallname{align}%
    \else
      \def\doitallname{multline}%
    \fi
  \else
    \def\doitallname{equation}
  \fi
  \begingroup\edef\x{\endgroup
    \noexpand\begin{\doitallname\doitallstar}%
    \noexpand\BODY
    \noexpand\end{\doitallname\doitallstar}%
  }\x
}
\def\[#1\]{\begin{doitall}#1\end{doitall}}

\newcommand{\newreptheorem}[2]{\newtheorem*{rep@#1}{\rep@title}\newenvironment{rep#1}[1]{\def\rep@title{\bf #2 \ref*{##1}}\begin{rep@#1}}{\end{rep@#1}}}

\makeatother

\newreptheorem{theorem}{Theorem}
\newreptheorem{lemma}{Lemma}
\newreptheorem{corollary}{Corollary}

\newtheorem*{rep@thm}{\rep@title} \newcommand{\newrepthm}[2]{%
\newenvironment{rep#1}[1]{%
\def\rep@title{\autoref{##1}}%
\begin{rep@thm} }%
{\end{rep@thm} } }
\makeatother
\newrepthm{thm}{}
\newrepthm{lem}{}
\newrepthm{crl}{}

\usepackage{framed}

\usepackage{color,tikz,environ}
\usetikzlibrary{decorations.markings,arrows,plotmarks}
\tikzset{nomorepostaction/.code={\let\tikz@postactions\pgfutil@empty}}

\tikzset{middlearrow/.style={
        decoration={markings,
            mark= at position 0.5 with {\arrow{#1}} ,
        },
        postaction={decorate}
    }
}

\tikzset{onethirdarrow/.style={
        decoration={markings,
            mark= at position 0.33 with {\arrow{#1}} ,
        },
        postaction={decorate}
    }
}

\tikzset{twothirdarrow/.style={
        decoration={markings,
            mark= at position 0.67 with {\arrow{#1}} ,
        },
        postaction={decorate}
    }
}

\tikzset{endarrow/.style={
        decoration={markings,
            mark= at position 0.9 with {\arrow{#1}} ,
        },
        postaction={decorate}
    }
}

\tikzset{startarrow/.style={
        decoration={markings,
            mark= at position 0.1 with {\arrow{#1}} ,
        },
        postaction={decorate}
    }
}

\date{}
\title{A Tight Lower Bound for \edp \\on Planar DAGs\thanks{A preliminary version of this paper appeared in CIAC 2021.}}

\author{Rajesh~Chitnis%\thanks{Work done while at the University of Warwick, UK and supported by ERC grant 2014-CoG 647557.}
}
\affil{School of Computer Science, University of Birmingham, UK.\\
\texttt{rajeshchitnis@gmail.com}}
\usepackage{graphicx}
\usepackage{xspace,amssymb,amsmath,amsfonts}
\usepackage{multirow}

\usepackage{hyperref}
\usepackage{tikz,environ}
\usetikzlibrary{decorations.pathreplacing}
\usetikzlibrary{decorations.markings,arrows}
\tikzset{middlearrow/.style={
        decoration={markings,
            mark= at position 0.5 with {\arrow{#1}} ,
        },
        postaction={decorate}
    }
}

\usepackage[strings]{underscore}

\AtBeginDocument{\def\sectionautorefname{Section}}
\AtBeginDocument{\def\subsectionautorefname{Section}}
\AtBeginDocument{\def\subsubsectionautorefname{Section}}
\definecolor{darkblue}{rgb}{0,0,1}
\definecolor{darkred}{rgb}{0.6,0,0}
\definecolor{darkgreen}{rgb}{0,1,0}
\hypersetup{colorlinks, linkcolor=darkblue, citecolor=darkblue,
urlcolor=darkblue}

\newcommand{\gt}{\textsc{Grid-Tiling}\xspace}
\newcommand{\gtleq}{\textsc{Grid-Tiling-}$\leq$\xspace}
\newcommand{\disjp}{\textsc{Disjoint Paths}\xspace}
\newcommand{\edp}{\textsc{Edge-Disjoint Paths}\xspace}
\newcommand{\vdp}{\textsc{Vertex-Disjoint Paths}\xspace}

\newcommand{\Le}{\texttt{Left}}
\newcommand{\Ri}{\texttt{Right}}
\newcommand{\To}{\texttt{Top}}
\newcommand{\Bo}{\texttt{Bottom}}
\newcommand{\LB}{\text{LB}}
\newcommand{\TR}{\text{TR}}
\newcommand{\lefty}{\texttt{west}}
\newcommand{\righty}{\texttt{east}}
\newcommand{\topy}{\texttt{north}}
\newcommand{\bottomy}{\texttt{south}}
\newcommand{\poly}{\text{poly}}
\newcommand{\w}{\textbf{w}}
\newcommand{\x}{\textbf{x}}
\newcommand{\h}{\textbf{h}}
\newcommand{\ve}{\textbf{v}}
\newcommand{\Path}{\texttt{Path}}
\newcommand{\Matching}{\texttt{Matching}}
\newcommand{\Sink}{\texttt{Sink}}
\newcommand{\Source}{\texttt{Source}}
\newcommand{\Row}{\texttt{RowPath}}
\newcommand{\Column}{\texttt{ColumnPath}}
\newcommand{\z}{\textbf{z}}
\newcommand{\Horizontal}{\textsc{Horizontal}}
\newcommand{\Vertical}{\textsc{Vertical}}
\newcommand{\splitt}{\texttt{split}}

\begin{document}
\renewcommand*{\sectionautorefname}{Section}
\renewcommand*{\subsectionautorefname}{Section}
\renewcommand*{\subsubsectionautorefname}{Section}

\maketitle

\begin{abstract}
Given a graph $G$ and a set $\mathcal{T}=\big\{ (s_i, t_i) : 1\leq i\leq k \big\}$ of $k$ pairs, the \vdp (resp. \edp) problems asks to determine whether there exist pairwise vertex-disjoint (resp. edge-disjoint) paths $P_1, P_2, \ldots, P_k$ in $G$ such that $P_i$ connects $s_i$ to $t_i$ for each $1\leq i\leq k$. Unlike their undirected counterparts which are FPT (parameterized by $k$) from Graph Minor theory, both the edge-disjoint and vertex-disjoint versions in directed graphs were shown by Fortune et al. (TCS '80) to be NP-hard for $k=2$.
%As part of their seminal work on Graph Minors, Robertson and Seymour (JCTB '95) showed that both \edp and \vdp are FPT parameterized by $k$ on undirected graphs. Fortune et al. (TCS '80) showed that checking for the existence of disjoint paths is significantly harder in directed graphs: both \edp and \vdp are NP-hard even for $k=2$.
%\quad
This strong hardness for \disjp on general directed graphs led to the study of parameterized complexity on special graph classes, e.g., when the underlying undirected graph is planar. For \vdp on planar directed graphs, Schrijver (SICOMP '94) designed an $n^{O(k)}$ time algorithm which was later improved upon by Cygan et al. (FOCS '13) who designed an FPT algorithm running in $2^{2^{O(k^2)}}\cdot n^{O(1)}$ time. To the best of our knowledge, the parameterized complexity of \edp on planar\footnote{A directed graph is planar if its underlying undirected graph is planar.} directed graphs is unknown.

\quad We resolve this gap by showing that \edp is W[1]-hard parameterized by the number $k$ of terminal pairs, even when the input graph is a planar directed acyclic graph (DAG). This answers a question of Slivkins (ESA '03, SIDMA '10). Moreover, under the Exponential Time Hypothesis (ETH), we show that there is no $f(k)\cdot n^{o(k)}$ algorithm for \edp on planar DAGs, where $k$ is the number of terminal pairs, $n$ is the number of vertices and $f$ is any computable function. Our hardness holds even if both the maximum in-degree and maximum out-degree of the graph are at most $2$.

We now place our result in the context of previously known algorithms and hardness for \edp on special classes of directed graphs:
\begin{itemize}
  \item \textbf{Implications for \edp on DAGs}: Our result shows that the $n^{O(k)}$ algorithm of Fortune et al. (TCS '80) for \edp on DAGs is asymptotically tight, even if we add an extra restriction of planarity. The previous best lower bound (also under ETH) for \edp on DAGs was $f(k)\cdot n^{o(k/\log k)}$ by Amiri et al. (MFCS '16, IPL '19) which improved upon the $f(k)\cdot n^{o(\sqrt{k})}$ lower bound implicit in Slivkins (ESA '03, SIDMA '10).
  \item \textbf{Implications for \edp on planar directed graphs}: As a special case of our result, we obtain that \edp on planar directed graphs is W[1]-hard parameterized by the number $k$ of terminal pairs. This answers a question of Cygan et al. (FOCS '13) and Schrijver (pp. 417-444, Building Bridges II, '19), and completes the landscape (see~\autoref{table:disjoint-paths-landscape}) of the parameterized complexity status of edge and vertex versions of the \disjp problem on planar directed and planar undirected graphs.
\end{itemize}

%\keywords{First keyword  \and Second keyword \and Another keyword.}
\end{abstract}

\section{Introduction}
\label{sec:intro}

%\todo{Basic intro of disjoint paths}

The \disjp problem is one of the most fundamental problems in graph theory: given a graph and a set of $k$ terminal pairs, the question is to determine whether there exists a collection of $k$ pairwise disjoint paths where each path connects one of the given terminal pairs. There are four natural variants of this problem depending on whether we consider undirected or directed graphs and the edge-disjoint or vertex-disjoint requirement. In undirected graphs, the edge-disjoint version is reducible to the vertex-disjoint version in polynomial time by considering the line graph. In directed graphs, the edge-disjoint version and vertex-disjoint version are known to be equivalent in terms of designing exact algorithms. Besides its theoretical importance, the \disjp problem has found applications in VLSI design, routing, etc. The interested reader is referred to the surveys~\cite{survey-1} and ~\citep[Chapter 9]{survey-2} for more details.

The case when the number of terminal pairs $k$ are bounded is of special interest: given a graph with $n$ vertices and $k$ terminal pairs the goal is to try to design either FPT algorithms, i.e., algorithms whose running time is $f(k)\cdot n^{O(1)}$ for some computable function $f$, or XP algorithms, i.e., algorithms whose running time is $n^{g(k)}$ for some computable function $g$. We now discuss some of the known results on exact\footnote{This paper focuses on exact algorithms for the \disjp problem so we do not discuss here the results regarding (in)approximability.
%There are also several results in the literature related to (in)approximability of the \disjp problem, but we do not discuss them here since our focus is on exact algorithms in this paper.
} algorithms for different variants of the \disjp problem before stating our result.

%The case when the number of terminal pairs $k$ are bounded is of special interest: can one design FPT algorithms, i.e., algorithms whose running time is $f(k)\cdot n^{O(1)}$ for some computable function $f$? For undirected graphs, an FPT algorithm~\cite{DBLP:journals/jct/RobertsonS95b} parameterized by $k$ for the \disjp problem was designed by Robertson and Seymour and it forms a core algorithmic component of their seminal work on the theory of Graph Minors.

%\paragraph*{Prior work on exact algorithms for \disjp on undirected graphs}
\paragraph*{Prior work on exact algorithms for \disjp on undirected graphs:}
The NP-hardness for \edp and \vdp on undirected graphs was shown by Even et al.~\cite{DBLP:conf/focs/EvenIS75}. Solving the \vdp problem on undirected graphs is an important subroutine in checking whether a fixed graph $H$ is a minor of a graph $G$. Hence, a core algorithmic result of the seminal work of Robertson and Seymour was their FPT algorithm~\cite{DBLP:journals/jct/RobertsonS95b} for \vdp (and hence also \edp) on general undirected graphs which runs in $O(g(k)\cdot n^{3})$ time for some function $g$. The cubic dependence on the input size was improved to quadratic by Kawarabayashi et al.~\cite{DBLP:journals/jct/KawarabayashiKR12} who designed an algorithm running in $O(h(k)\cdot n^{2})$ time for some function $h$. Both the functions $g$ and $h$ are quite large (at least quintuple exponential as per~\cite{DBLP:journals/jct/AdlerKKLST17}). This naturally led to the search for faster FPT algorithms on planar graphs: Adler et al.~\cite{DBLP:journals/jct/AdlerKKLST17} designed an algorithm for \vdp on planar graphs which runs in $2^{2^{O(k^2)}}\cdot n^{O(1)}$ time. Very recently, this was improved to an single-exponential time FPT algorithm which runs in $2^{O(k^2)}\cdot n^{O(1)}$ time by Lokshtanov et al.~\cite{DBLP:conf/stoc/LokshtanovMP0Z20}.

There are two more variants of the \disjp problem: the \emph{half-integral} version where each vertex/edge can belong to at most two paths, and the \emph{parity} version where the length of each path is required to respect a given parity (even or odd) condition.
%Kleinberg~\cite{DBLP:conf/stoc/Kleinberg98} gave an $O_{k}(n^{3})$ time algorithm\footnote{We use the $O_{k}(\cdot)$ notation to hide factors depending only on $k$} for the half-integral \vdp problem on general undirected graphs, which was improved to $O_{k}(n \log n)$ by Kawarabayashi and Reed~\cite{DBLP:conf/soda/KawarabayashiR08}.
FPT algorithms are known for each of the following versions of \vdp on general undirected graphs: the half-integral version~\cite{DBLP:conf/stoc/Kleinberg98,DBLP:conf/soda/KawarabayashiR08}, the half-integral version with parity~\cite{DBLP:conf/soda/KawarabayashiR09} and finally just the parity version (without half-integral)~\cite{DBLP:conf/focs/KawarabayashiRW11}.

%\disjp by Kawarabayashi and Reed~\cite{DBLP:conf/soda/KawarabayashiR08} and later improved to parity and half integral~\cite{DBLP:conf/soda/KawarabayashiR09}. Finally, just the parity version (without half-integral) of \disjp by Kawarabyashi et al.~\cite{DBLP:conf/focs/KawarabayashiRW11}.
%

%\paragraph*{Prior work on exact algorithms for \disjp on directed graphs}

%~\\

%\subsection*{Prior work on exact algorithms for \disjp on directed graphs}
\paragraph*{Prior work on exact algorithms for \disjp on directed graphs:}
Unlike undirected graphs where both \edp and \vdp are FPT parameterized by $k$, the \disjp problem becomes significantly harder for directed graphs: Fortune et al.~\cite{DBLP:journals/tcs/FortuneHW80} showed that both \edp and \vdp on general directed graphs are NP-hard even for $k=2$. For general directed graphs, Giannopoulou et al.~\cite{kreutzer-new} recently designed an XP algorithm for the half-integral version of \disjp: here the goal is to either find a set of $k$ paths $P_{1},P_{2},\ldots, P_{k}$ such that $P_{i}$ is an $s_i \leadsto t_i$ path for each $i\in [k]$ and each vertex in the graph appears in at most two of the paths, or conclude that the given instance has no solution with pairwise disjoint paths. This algorithm improves upon an older XP algorithm of Kawarabayashi et al.~\cite{kreutzer-old} for the quarter-integral case in general digraphs.

The \disjp problem has also been extensively studied on special subclasses of digraphs:
%\todo{Add references to work of Kreutzer et al.}\\
%\begin{itemize}
%\item
\begin{itemize}
  \item \textbf{\disjp on DAGs}: It is easy to show that \vdp and \edp are equivalent on the class of directed acyclic graphs (DAGs). Fortune et al.~\cite{DBLP:journals/tcs/FortuneHW80} designed an $n^{O(k)}$ algorithm for \edp on DAGs. Slivkins~\cite{DBLP:journals/siamdm/Slivkins10} showed W[1]-hardness for \edp on DAGs and a $f(k)\cdot n^{o(\sqrt{k})}$ lower bound (for any computable function $f$) under the Exponential Time Hypothesis~\cite{eth,eth-2} (ETH) follows from that reduction. Amiri et al.~\cite{DBLP:journals/ipl/AmiriKMR19}\footnote{We note that~\cite{DBLP:journals/ipl/AmiriKMR19} considers a more general version than \disjp which allows congestion} improved the lower bound to $f(k)\cdot n^{o(k/\log k)}$ thus showing that the algorithm of Fortune et al.~\cite{DBLP:journals/tcs/FortuneHW80} is almost-tight.

  \item \textbf{\disjp on directed planar graphs}: Schrijver~\cite{DBLP:journals/siamcomp/Schrijver94} designed an $n^{O(k)}$ algorithm for \vdp on directed planar graphs. This was improved upon by Cygan et al.~\cite{DBLP:conf/focs/CyganMPP13} who designed an FPT algorithm running in $2^{2^{O(k^2)}}\cdot n^{O(1)}$ time. As pointed out by Cygan et al.~\cite{DBLP:conf/focs/CyganMPP13}, their FPT algorithm for \vdp on directed planar graphs does not work for the \edp problem. The status of parameterized complexity (parameterized by $k$) of \edp on directed planar graphs remained an open question.~\autoref{table:disjoint-paths-on-directed-graphs} gives a summary of known results for exact algorithms for \disjp on (subclasses of) directed graphs.
\end{itemize}

\begin{table}[!h]
\begin{tabular}{ |c|c|c|c| }
%\hline
%\multicolumn{3}{ |c| }{Parameterized Complexity of \textsc{Disjoint Paths} on directed graphs} \\
\hline
\textbf{Graph class} & \textbf{Problem type} & \textbf{Algorithm} & \textbf{Lower Bound}  \\ \hline \hline
General graphs & Vertex-disjoint = edge-disjoint & ???? & NP-hard for $k=2$\\ \hline
\multirow{3}{*}{DAGs} & \multirow{3}{*}{Vertex-disjoint = edge-disjoint} & \multirow{3}{*}{$n^{O(k)}$~\cite{DBLP:journals/tcs/FortuneHW80}} & $f(k)\cdot n^{o(\sqrt{k})}$~\cite{DBLP:journals/siamdm/Slivkins10}\\
 &  &  & $f(k)\cdot n^{o(k/\log k)}$~\cite{DBLP:journals/ipl/AmiriKMR19}\\
 & &  & $f(k)\cdot n^{o(k)}$ \textbf{[this paper]}\\ \hline
\multirow{3}{*}{Planar graphs} & \multirow{2}{*}{Vertex-disjoint} & $n^{O(k)}$~\cite{DBLP:journals/siamcomp/Schrijver94} & \multirow{2}{*}{????} \\
 &  & $2^{2^{O(k^2)}}\cdot n^{O(1)}$~\cite{DBLP:conf/focs/CyganMPP13} & \\\cline{2-4}
 & Edge-disjoint & ???? & $f(k)\cdot n^{o(k)}$ \textbf{[this paper]} \\ \hline
\multirow{2}{*}{Planar DAGs} & Vertex-disjoint & $2^{2^{O(k^2)}}\cdot n^{O(1)}$~\cite{DBLP:conf/focs/CyganMPP13} & ????\\
 & Edge-disjoint & $n^{O(k)}$~\cite{DBLP:journals/tcs/FortuneHW80} & $f(k)\cdot n^{o(k)}$ \textbf{[this paper]} \\
\hline
\end{tabular}
~\\
\caption{The landscape of parameterized complexity results for \disjp on directed graphs. All lower bounds are under the Exponential Time Hypothesis (ETH). To the best of our knowledge, the entries marked with ???? have no known non-trivial results.
\label{table:disjoint-paths-on-directed-graphs}
}
\end{table}

\paragraph*{Our result:}
%\subsection*{Our result}
We resolve this open question by showing a slightly stronger result: the \edp problem is W[1]-hard parameterized by $k$ when the input graph is a planar DAG whose max in-degree and max out-degree are both at most $2$. First we define the \edp problem formally below, and then state our result:

\begin{center}
\noindent\framebox{\begin{minipage}{6in}
\textbf{\edp}\\
\emph{\underline{Input}}: A directed graph $G=(V,E)$, and a set
$\mathcal{T}\subseteq V\times V$ of $k$ terminal pairs given by $\big\{(s_i, t_i) : 1\leq i\leq k\big\}$. \\
\emph{\underline{Question}}: Do there exist $k$ pairwise edge-disjoint paths $P_1, P_2, \ldots, P_k$
such that $P_i$ is an $s_i \leadsto t_i$ path for each $1\leq i\leq k$?\\
\emph{\underline{Parameter}}: $k$
\end{minipage}}
\end{center}

%\todo{State that theorem holds even if in-degree and out-degree are both upper bounded by 2}

\begin{restatable}{theorem}{mainthm}
%\normalfont
The \edp problem on planar DAGs is W[1]-hard parameterized by the number $k$ of terminal pairs. Moreover, under ETH, the \edp problem on planar DAGs cannot be solved $f(k)\cdot n^{o(k)}$  time where $f$ is any computable function, $n$ is the number of vertices and $k$ is the number of terminal pairs. The hardness holds even if both the maximum in-degree and maximum out-degree of the graph are at most $2$.
\label{thm:main-result}
\end{restatable}

\noindent Recall that the Exponential Time Hypothesis (ETH) states that $n$-variable $m$-clause 3-SAT cannot be solved in $2^{o(n)}\cdot (n+m)^{O(1)}$ time~\cite{eth,eth-2}. Prior to our result, only the NP-completeness of \edp on planar DAGs was known~\cite{DBLP:journals/dam/Vygen95}. The reduction used in~\autoref{thm:main-result} is heavily inspired by some known reductions: in particular, the planar DAG structure (\autoref{fig:main}) is from~\cite{rajesh-soda-14,DBLP:journals/siamcomp/ChitnisFHM20} and the splitting operation (\autoref{fig:split} and~\autoref{defn:splitting-operation}) is from~\cite{DBLP:conf/csr/ChitnisF18,DBLP:journals/algorithmica/ChitnisFS19}.
We view the simplicity of our reduction as evidence of success of the (now) established methodology of showing W[1]-hardness (and ETH-based hardness) for planar graph problems using \gt and its variants.

\paragraph*{Placing~\autoref{thm:main-result} in the context of prior work:}
Theorem~\ref{thm:main-result} answers a question of Slivkins~\cite{DBLP:journals/siamdm/Slivkins10} regarding the parameterized complexity of \edp on planar DAGs. As a special case of~\autoref{thm:main-result}, one obtains that \edp on planar directed graphs is W[1]-hard parameterized by the number $k$ of terminal pairs: this answers a question of Cygan et al.~\cite{DBLP:conf/focs/CyganMPP13} and Schrijver~\cite{schrijver-building-bridges-II}. The W[1]-hardness result of~\autoref{thm:main-result} completes the landscape (see~\autoref{table:disjoint-paths-landscape}) of parameterized complexity of edge-disjoint and vertex-disjoint versions of the \disjp problem on planar directed and planar undirected graphs.~\autoref{thm:main-result} also shows that the $n^{O(k)}$ algorithm of Fortune et al.~\cite{DBLP:journals/tcs/FortuneHW80} for \edp on DAGs is asymptotically optimal, even if we add an extra restriction of planarity to the mix.~\autoref{thm:main-result} adds another problem (\edp on DAGs) to the relatively small list of problems for which it is provably known that the planar version has the same asymptotic complexity as the problem on general graphs: the only such other problems we are aware of are~\cite{DBLP:journals/algorithmica/ChitnisFS19,DBLP:journals/siamcomp/ChitnisFHM20,DBLP:conf/focs/MarxPP18}.
This is in contrast to the fact that for several problems~\cite{DBLP:conf/icalp/Marx12,DBLP:conf/icalp/KleinM12,DBLP:conf/focs/MarxPP18,KleinM14,FominLMPPS16,DemaineFHT05,DBLP:conf/stacs/PilipczukPSL13,DBLP:conf/esa/MarxP15,DBLP:conf/fsttcs/LokshtanovSW12,DBLP:journals/corr/AboulkerBHMT15,FominKLPS16}.  the planar version is easier by (roughly) a square root factor in the exponent as compared to general graphs, and there are lower bounds indicating that this improvement is essentially the best possible~\cite{DBLP:conf/icalp/Marx13}.

\begin{table}[!h]
\begin{tabular}{ |c|c|c| }
%\hline
%\multicolumn{3}{ |c| }{Parameterized Complexity of \textsc{Disjoint Paths} on directed graphs} \\
\hline
\textbf{Graph class} & \textbf{Problem type} & \textbf{Parameterized Complexity parameterized by $k$}    \\ \hline \hline
\multirow{2}{*}{Planar undirected} & Vertex-disjoint & \multirow{2}{*}{FPT~\cite{DBLP:journals/jct/AdlerKKLST17,DBLP:conf/stoc/LokshtanovMP0Z20,DBLP:journals/jct/KawarabayashiKR12,DBLP:journals/jct/RobertsonS95b}} \\ \cline{2-2}
 & Edge-disjoint & \\ \hline
\multirow{2}{*}{Planar directed} & Vertex-disjoint & FPT~\cite{DBLP:conf/focs/CyganMPP13} \\ \cline{2-3}
 & Edge-disjoint & W[1]-hard \textbf{[this paper]}\\ \hline
\end{tabular}
~\\
\caption{The landscape of parameterized complexity results for the four different versions (edge-disjoint vs vertex-disjoint \& directed vs undirected) of  \disjp on planar graphs.
\label{table:disjoint-paths-landscape}
}
\end{table}

\paragraph*{Organization of the paper:} In~\autoref{subsec:construction} we describe the construction of the instance $(G_2, \mathcal{T})$ of \edp. The two directions of the reduction are shown in~\autoref{subsec:hard-direction} and~\autoref{subsec:easy-direction} respectively. Finally,~\autoref{subsec:proof-of-main-theorem} contains the proof of~\autoref{thm:main-result}. We conclude with some open questions in~\autoref{sec:conclusion}.
%The notation used in this paper is defined in~\autoref{secapp:notation}.

\paragraph*{Notation:}
\label{secapp:notation}

All graphs considered in this paper are directed and do not have self-loops or multiple edges. We use (mostly) standard graph theory notation~\cite{diestel-book}. The set $\{1,2,3,\ldots, M\}$ is denoted by $[M]$ for each $M\in \mathbb{N}$. A directed edge (resp. path) from $s$ to $t$ is denoted by $s\to t$ (resp. $s\leadsto t$). %A directed path from $s$ to $t$ is denoted by $s\leadsto t$.
 %A self-loop on the vertex $s$ is denoted by $s\to s$: this is not a problem for edge-disjointness as self-loops can just be removed from any path.
 We use the \textbf{non-standard} notation (to avoid having to consider different cases in our proofs): $s\leadsto s$ \textbf{does not} represent a self-loop but rather is to be viewed as \emph{``just staying put"} at the vertex $s$. If $A,B\subseteq V(G)$ then we say that there is an $A\leadsto B$ path if and only if there exists two vertices $a\in A, b\in B$ such that there is an $a\leadsto b$ path. For $A\subseteq V(G)$ we define $N_{G}^{+}(A) = \big\{ x\notin A\ : \exists\ y\in A\ \text{such that } (y,x)\in E(G) \big\}$ and $N_{G}^{-}(A) = \big\{ x\notin A\ : \exists\ y\in A\ \text{such that } (x,y)\in E(G) \big\}$. For $A\subseteq V(G)$ we define $G[A]$ to be the graph induced on the vertex set $A$, i.e., $G[A]:= (A,E_A)$ where $E_A:=E(G)\cap (A\times A)$.

\section{W[1]-hardness of \edp on Planar DAGs}
\label{sec:main-lb}

To obtain W[1]-hardness for \edp on planar DAGs, we reduce from the \gtleq problem~\cite{DBLP:conf/compgeom/MarxS14} which is defined below:
\begin{center}
\noindent\framebox{\begin{minipage}{6in}
%\todo{Maybe indices need to be swapped! Check after the proof is over}\\
\textbf{\gtleq}\\
\emph{\underline{Input}}: Integers $k, N$, and a collection $\mathcal{S}$ of $k^2$ sets given by $\big\{ S_{x,y}\subseteq
[N]\times [N]\ : 1\leq x, y\leq k \big\}$.\\
\emph{\underline{Question}}: For each $1\leq x, y\leq k$ does there exist a pair
$\gamma_{x,y}\in S_{x,y}$ such that
\begin{itemize}
\item if $\gamma_{x,y}=(a,b)$ and $\gamma_{x+1,y}=(a',b')$ then $b\leq b'$, and
\item if $\gamma_{x,y}=(a,b)$ and $\gamma_{x,y+1}=(a',b')$ then $a\leq a'$
\end{itemize}
\end{minipage}}
\end{center}

%\begin{center}
%\begin{figure}[!h]
%%\centering
%\vspace{-5mm}
%\includegraphics[width=6in]{figs/gtleq}
%\vspace{-45mm}
%
%\caption{An instance of \gtleq with $k=3, N=5$ and a solution highlighted in red.
%Note that in a solution, the second coordinates in a row are non-decreasing as we go from left to right and the first coordinates in a column are non-decreasing as we go from bottom to top.
%%By non-decreasing we mean that the next entry is greater equal the entry coming before it.
%\label{fig:gtleq}
%}
%\end{figure}
%\end{center}

\begin{center}
\begin{figure}[!h]
%\centering
\vspace{-5mm}
\includegraphics[width=6in]{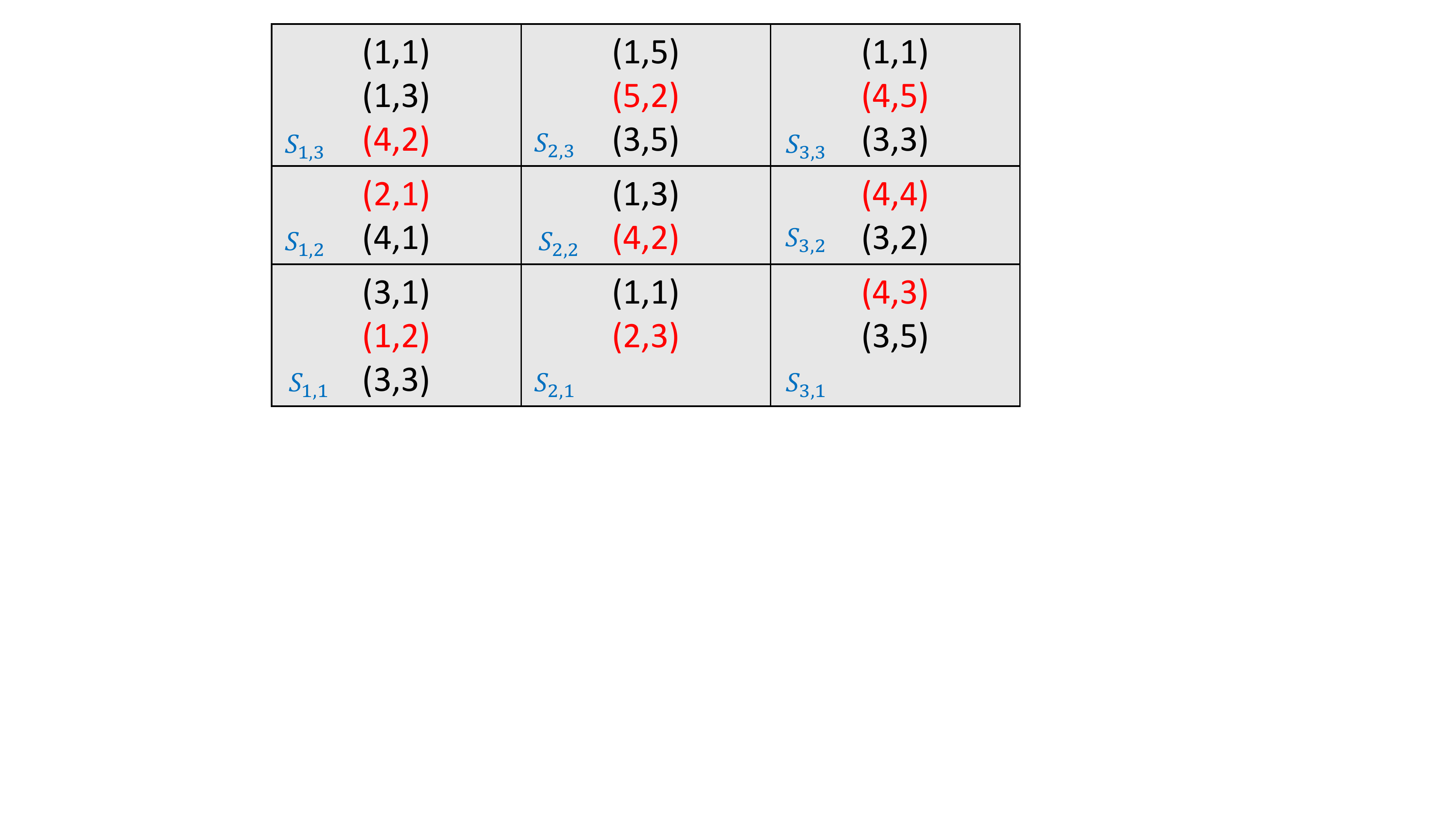}
\vspace{-45mm}

\caption{An instance of \gtleq with $k=3, N=5$ and a solution highlighted in red.
Note that in a solution, the second coordinates in a row are non-decreasing as we go from left to right and the first coordinates in a column are non-decreasing as we go from bottom to top.
%By non-decreasing we mean that the next entry is greater equal the entry coming before it.
\label{fig:gtleq}
}
\end{figure}
\end{center}

~\autoref{fig:gtleq} gives an illustration of an instance of \gtleq along with a solution. It is known~\citep[Theorem 14.30]{fpt-book} that \gtleq is W[1]-hard parameterized by $k$, and under the Exponential Time Hypothesis (ETH) has no $f(k)\cdot N^{o(k)}$ algorithm for any computable function $f$. We will exploit this result by reducing an instance $(k,N,\mathcal{S})$ of \gtleq in $\poly(N,k)$ time to an instance $(G_2, \mathcal{T})$ of \edp such that $G_2$ is a planar DAG, number of vertices in $G_2$ is $|V(G_2)|=O(N^2 k^2)$ and number of terminal pairs is $|\mathcal{T}|=2k$.

\begin{remark}
\label{remark:issue-of-gtleq-swapping-indices}
\normalfont
Our definition of \gtleq above is slightly different than the one given in~\citep[Theorem 14.30]{fpt-book}: there the constraints are first coordinate of $\gamma_{x,y}$ is $\leq$ first coordinate of $\gamma_{x+1,y}$ and second coordinate of $\gamma_{x,y}$ is $\leq$ second coordinate of $\gamma_{x,y+1}$. By rotating the axis by $90^{\circ}$, i.e., swapping the indices, our version of \gtleq is equivalent to that from~\citep[Theorem 14.30]{fpt-book}.
\end{remark}

%\todo{Include a figure for \gtleq instance\footnote{a}}

%\todo{This section is organized as follows: in Section 2.1 we describe the construction of the instance $(G_2, \mathcal{T})$ of \edp. The two directions of the reduction are shown in Section 2.2 and Section 2.3 respectively. Finally, Section 2.4 contains the proof of Theorem}

%The rest of this section is organized as follows: in~\autoref{subsec:construction} we describe the construction of the instance $(G_2, \mathcal{T})$ of \edp. The two directions of the reduction are shown in~\autoref{subsec:hard-direction} and~\autoref{subsec:easy-direction} respectively. Finally,~\autoref{subsec:proof-of-main-theorem} contains the proof of~\autoref{thm:main-result}.

\subsection{Construction of the instance $(G_2,\mathcal{T})$ of \edp}
\label{subsec:construction}

%\todo{This section is organized as follows: in Section 2.1.1 we describe the construction of an intermediate graph $G_1$ (\autoref{fig:main}). The splitting operation is defined in Section 2.1.2, and the graph $G_2$ is obtained from $G_1$ by splitting each (black) grid vertex}

Consider an instance $(N,k,\mathcal{S})$ of \gtleq. We now build an instance $(G_2,\mathcal{T})$ of \edp as follows: first in Section 2.1.1 we describe the construction of an intermediate graph $G_1$ (\autoref{fig:main}). The splitting operation is defined in Section 2.1.2, and the graph $G_2$ is obtained from $G_1$ by splitting each (black) grid vertex.

% (refer to~\autoref{fig:main}):
%
%This section is organized as follows: in Section 2.1.1 we describe the construction of an intermediate graph $G_1$ (\autoref{fig:main}). The splitting operation is defined in Section 2.1.2, and the graph $G_2$ is obtained from $G_1$ by splitting each (black) grid vertex.

\subsubsection{Construction of the graph $G_1$}
\label{subsubsec:construction-of-G1}

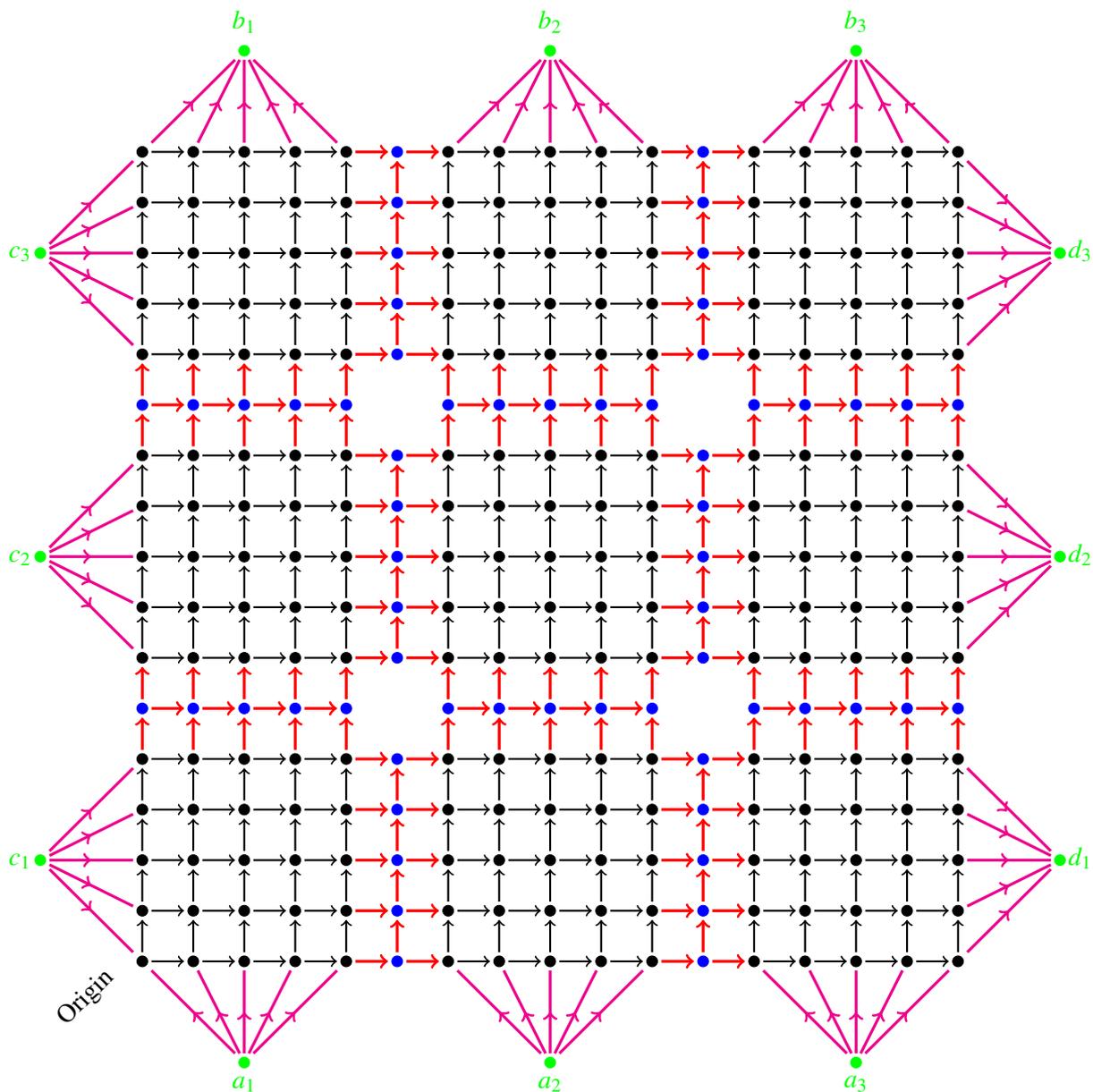
\begin{figure}
\centering
%\begin{scaletikzpicturetowidth}{\textwidth}
\begin{tikzpicture}[scale=0.75]

%%%% Grids

\foreach \i in {0,1,2}
    \foreach \j in {0,1,2}
{
\begin{scope}[shift={(6*\i,6*\j)}]

        \foreach \x in {1,2,...,5}
        \foreach \y in {1,2,...,5}
    {
        \draw [black] plot [only marks, mark size=3, mark=*] coordinates {(\x,\y)};
    }

        \foreach \x in {1,2,...,5}
    \foreach \y in {1,2,3,4}
    {
        \path (\x,\y) node(a) {} (\x,\y+1) node(b) {};
        \draw[thick,->] (a) -- (b);
    }

        \foreach \y in {1,2,...,5}
        \foreach \x in {1,2,3,4}
    {
        \path (\x,\y) node(a) {} (\x+1,\y) node(b) {};
        \draw[thick,->] (a) -- (b);
    }

\end{scope}
}

%%%%%%%%%%%%% vertical connecting edges
\foreach \i in {0,1,2}
\foreach \j in {1,2}
{
\begin{scope}[shift={(6*\i,6*\j-6)}]

\foreach \x in {1,2,...,5}
{
    \draw [blue] plot [only marks, mark size=3, mark=*] coordinates {(\x,6)};
    \path (\x,5) node(a) {} (\x,6) node(b) {};
        \draw[red,very thick,->] (a) -- (b);
    \path (\x,6) node(a) {} (\x,7) node(b) {};
        \draw[red,very thick,->] (a) -- (b);
}

\foreach \x in {1,2,...,4}
{
    \path (\x,6) node(a) {} (\x+1,6) node(b) {};
        \draw[red,very thick,->] (a) -- (b);
}
\end{scope}
}

%%%%%%%%%%%%% horizontal connecting edges
\foreach \j in {0,1,2}
\foreach \i in {0,1}
{
\begin{scope}[shift={(6*\i,6*\j)}]

\foreach \y in {1,2,...,5}
{
    \draw [blue] plot [only marks, mark size=3, mark=*] coordinates {(6,\y)};
    \path (5,\y) node(a) {} (6,\y) node(b) {};
        \draw[red,very thick,->] (a) -- (b);
    \path (6,\y) node(a) {} (7,\y) node(b) {};
        \draw[red,very thick,->] (a) -- (b);
}

\foreach \y in {1,2,...,4}
{
    \path (6,\y) node(a) {} (6,\y+1) node(b) {};
        \draw[red,very thick,->] (a) -- (b);
}
\end{scope}
}

\draw [green] plot [only marks, mark size=3, mark=*] coordinates {(-1,3)}
node[label={[xshift=-3mm,yshift=-4mm] $c_{1}$}] {} ;

\draw [green] plot [only marks, mark size=3, mark=*] coordinates {(-1,9)}
node[label={[xshift=-3mm,yshift=-4mm] $c_{2}$}] {} ;

\draw [green] plot [only marks, mark size=3, mark=*] coordinates {(-1,15)}
node[label={[xshift=-3mm,yshift=-4mm] $c_{3}$}] {} ;

\draw [green] plot [only marks, mark size=3, mark=*] coordinates {(19,3)}
node[label={[xshift=3mm,yshift=-4mm] $d_{1}$}] {} ;

\draw [green] plot [only marks, mark size=3, mark=*] coordinates {(19,9)}
node[label={[xshift=3mm,yshift=-4mm] $d_{2}$}] {} ;

\draw [green] plot [only marks, mark size=3, mark=*] coordinates {(19,15)}
node[label={[xshift=3mm,yshift=-4mm] $d_{3}$}] {} ;

\foreach \k in {0,1,2}
{
\begin{scope}[shift={(0,6*\k)}]
\foreach \y in {1,2,...,5}
    {
        \path (-1,3) node(a) {} (1,\y) node(b) {};
        \draw[magenta,very thick,middlearrow={>}] (a) -- (b);

        \path (17,\y) node(a) {} (19,3) node(b) {};
        \draw[magenta,very thick,middlearrow={>}] (a) -- (b);
    }
\end{scope}
}

%%%%%%% top & bottom terminal vertices

\draw [green] plot [only marks, mark size=3, mark=*] coordinates {(3,-1)}
node[label={[xshift=0mm,yshift=-7mm] $a_{1}$}] {} ;

\draw [green] plot [only marks, mark size=3, mark=*] coordinates {(9,-1)}
node[label={[xshift=0mm,yshift=-7mm] $a_{2}$}] {} ;

\draw [green] plot [only marks, mark size=3, mark=*] coordinates {(15,-1)}
node[label={[xshift=0mm,yshift=-7mm] $a_{3}$}] {} ;

\draw [green] plot [only marks, mark size=3, mark=*] coordinates {(3,19)}
node[label={[xshift=0mm,yshift=0mm] $b_{1}$}] {} ;

\draw [green] plot [only marks, mark size=3, mark=*] coordinates {(9,19)}
node[label={[xshift=0mm,yshift=0mm] $b_{2}$}] {} ;

\draw [green] plot [only marks, mark size=3, mark=*] coordinates {(15,19)}
node[label={[xshift=0mm,yshift=0mm] $b_{3}$}] {} ;

\foreach \k in {0,1,2}
{
\begin{scope}[shift={(6*\k,0)}]
\foreach \x in {1,2,...,5}
    {
        \path (3,-1) node(a) {} (\x,1) node(b) {};
        \draw[magenta,very thick,middlearrow={>}] (a) -- (b);

        \path (3,19) node(a) {} (\x,17) node(b) {};
        \draw[magenta,very thick,middlearrow={>}] (b) -- (a);
    }
\end{scope}
}

%%%%%%%%%%%%%%%     Marking origin
\draw [rotate=45,black] plot [only marks, mark size=0, mark=*] coordinates
{(0,0)}
node[label={[rotate=45,xshift=0mm,yshift=-2mm] Origin}] {} ;

\end{tikzpicture}
\caption{The graph $G_1$ constructed for the input $k=3$ and $N=5$ via the construction described in~\autoref{subsubsec:construction-of-G1}. The final graph $G_2$ for the \edp instance is obtained from $G_1$ by the splitting operation (\autoref{defn:splitting-operation}) as described in~\autoref{subsubsec:construction-of-G2}.
\label{fig:main}
}
\end{figure}

Given integers $k$ and $N$, we build a directed graph $G_1$ as follows (refer to~\autoref{fig:main}):
%\begin{bracketenumerate}
\begin{enumerate}
\item \textbf{Origin}: The origin is marked at the bottom left corner of~\autoref{fig:main}. This is defined just so we can view the naming of the vertices as per the usual $X-Y$ coordinate system: increasing horizontally towards the right, and vertically towards the top.

\item \textbf{Grid (black) vertices and edges}: For each $1\leq i,j\leq k$ we introduce a (directed) $N\times N$ grid $G_{i,j}$ where the column numbers increase from $1$ to $N$ as we go from left to right, and the row numbers increase from $1$ to $N$ as we go from bottom to top. For each $1\leq q,\ell\leq N$ the unique vertex which is the intersection of the $q^{\text{th}}$ column and $\ell^{\text{th}}$ row of $G_{i,j}$ is denoted by $\w_{i,j}^{q,\ell}$. The vertex set and edge set of $G_{i,j}$ is defined formally as:
%defined as follows:
\begin{itemize}
 %   \item In $G_{i,j}$ the column numbers increase from $1$ to $N$ as we go from left to right, and the row numbers increase from $1$ to $N$ as we go from bottom to top.

 %   \item For each $1\leq q,\ell\leq N$ the unique vertex which is the intersection of the $q^{\text{th}}$ column and $\ell^{\text{th}}$ row of $G_{i,j}$ is denoted by $\w_{i,j}^{q,\ell}$.

    \item $V(G_{i,j})= \big\{ \w_{i,j}^{q,\ell} : 1\leq q,\ell\leq N \big\}$

    %\item Each horizontal edge of the grid $G_{i,j}$ is oriented to the right, and each vertical edge is oriented towards the top.

    \item $E(G_{i,j}) = \Big(\bigcup_{(q,\ell)\in [N]\times [N-1]} \w_{i,j}^{q,\ell} \to \w_{i,j}^{q,\ell+1} \Big) \cup \Big( \bigcup_{(q,\ell)\in [N-1]\times [N]} \w_{i,j}^{q,\ell} \to \w_{i,j}^{q+1,\ell} \Big)$
\end{itemize}

All vertices and edges of $G_{i,j}$ are shown in~\autoref{fig:main} using black color. Note that each horizontal edge of the grid $G_{i,j}$ is oriented to the right, and each vertical edge is oriented towards the top. We will later (\autoref{defn:splitting-operation}) modify the grid $G_{i,j}$ to \emph{represent} the set $S_{i,j}$.
%The other grids serve to transfer information.
%In \autoref{fig:main} we highlight a gadget $G_{i,j}$ by a dotted rectangle.\\

For each $1\leq i,j\leq k$ we define the set of \emph{boundary} vertices of the grid $G_{i,j}$ as follows:
\begin{equation}\label{eqn:left-right-top-bottom-G1}
 %\left.
 \begin{aligned}
        \Le(G_{i,j}) := \big\{ \w_{i,j}^{1,\ell}\ :\ \ell\in [N]  \big\}\ ;\
        \Ri(G_{i,j}) := \big\{ \w_{i,j}^{N,\ell}\ :\ \ell\in [N]  \big\} \\
        \To(G_{i,j}) := \big\{ \w_{i,j}^{\ell,N}\ :\ \ell\in [N]  \big\}\ ;\
        \Bo(G_{i,j}) := \big\{ \w_{i,j}^{\ell,1}\ :\ \ell\in [N]  \big\}
       \end{aligned}
\end{equation}

\item \textbf{Arranging the $k^2$ different $N\times N$ grids $\{G_{i,j}\}_{1\leq i,j\leq k}$ into a large $k\times k$ grid}: We place the grids $G_{i,j}$ into a big $k\times k$ grid of grids left to right according to growing $i$ and from bottom to top according to growing $j$ (see the naming of the sets in~\autoref{fig:gtleq} in blue color). In particular,the grid $G_{1,1}$ is at bottom left corner of the construction, the grid $G_{k,k}$ at the top right corner, and so on.

\item \textbf{Blue vertices and red edges for horizontal connections}: For each $(i,j)\in [k-1]\times [k]$ we add a set of vertices $H_{i,j}^{i+1,j}:= \big\{ \h_{i,j}^{i+1,j}(\ell)\ :\ \ell\in [N] \big\}$ shown in~\autoref{fig:main} using \textcolor[rgb]{0.00,0.00,1.00}{blue} color. We also add the following three sets of edges (shown in~\autoref{fig:main} using \textcolor[rgb]{1.00,0.00,0.00}{red} color):
          \begin{itemize}
            \item a directed path of $N-1$ edges given by $\Path(H_{i,j}^{i+1,j}) := \big\{ \h_{i,j}^{i+1,j}(\ell)\to \h_{i,j}^{i+1,j}(\ell+1)\ :\ \ell\in [N-1] \big\}$
              \item a directed perfect matching from $\Ri(G_{i,j})$ to $H_{i,j}^{i+1,j}$ given by\\ $\Matching\big( G_{i,j}, H_{i,j}^{i+1,j} \big):= \big\{ \w_{i,j}^{N,\ell} \to \h_{i,j}^{i+1,j}(\ell)\ :\ \ell\in [N] \big\}$
              \item a directed perfect matching from $H_{i,j}^{i+1,j}$ to $\Le(G_{i+1,j})$ given by\\ $\Matching\big( H_{i,j}^{i+1,j}, G_{i+1,j} \big):= \big\{ \h_{i,j}^{i+1,j}(\ell) \to \w_{i+1,j}^{1,\ell}\ :\ \ell\in [N] \big\}$
            \end{itemize}

\item \textbf{Blue vertices and red edges for vertical connections}: For each $(i,j)\in [k]\times [k-1]$ we add a set of vertices $V_{i,j}^{i,j+1}:= \big\{ \ve_{i,j}^{i,j+1}(\ell)\ :\ \ell\in [N] \big\}$ shown in~\autoref{fig:main} using \textcolor[rgb]{0.00,0.00,1.00}{blue} color. We also add the following three sets of edges (shown in~\autoref{fig:main} using \textcolor[rgb]{1.00,0.00,0.00}{red} color):
          \begin{itemize}
            \item a directed path of $N-1$ edges given by $\Path(V_{i,j}^{i,j+1}) := \big\{ \ve_{i,j}^{i,j+1}(\ell)\to \ve_{i,j}^{i,j+1}(\ell+1)\ :\ \ell\in [N-1] \big\}$
              \item a directed perfect matching from $\To(G_{i,j})$ to $V_{i,j}^{i,j+1}$ given by\\ $\Matching\big( G_{i,j}, V_{i,j}^{i,j+1} \big):= \big\{ \w_{i,j}^{\ell,N} \to \ve_{i,j}^{i,j+1}(\ell)\ :\ \ell\in [N] \big\}$
              \item a directed perfect matching from $V_{i,j}^{i,j+1}$ to $\Bo(G_{i,j+1})$ given by\\ $\Matching\big( V_{i,j}^{i,j+1}, G_{i,j+1} \big):= \big\{ \ve_{i,j}^{i,j+1}(\ell) \to \w_{i,j+1}^{\ell,1}\ :\ \ell\in [N] \big\}$
            \end{itemize}

\item \textbf{Green (terminal) vertices and magenta edges}: For each $i\in [k]$ we add the following four sets of (terminal) vertices (shown in~\autoref{fig:main} using \textcolor[rgb]{0.00,1.00,.00}{green} color)
%\begin{itemize}
%  \item $A:= \big\{ a_i\ :\ i\in [k] \big\}$
%  \item $B:= \big\{ b_i\ :\ i\in [k] \big\}$
%  \item $C:= \big\{ c_i\ :\ i\in [k] \big\}$
%  \item $D:= \big\{ d_i\ :\ i\in [k] \big\}$
%\end{itemize}
%\begin{equation}
%A:= \big\{ a_i\ :\ i\in [k] \big\}\ ;\ B:= \big\{ b_i\ :\ i\in [k] \big\}\ ;\ C:= \big\{ c_i\ :\ i\in [k] \big\}\ ;\ D:= \big\{ d_i\ :\ i\in [k] \big\}
%\label{eqn:A-B-C-D}
%\end{equation}
\begin{equation}\label{eqn:A-B-C-D}
\begin{aligned}
        A := \big\{ a_i\ :\ i\in [k] \big\}\quad ;\quad B := \big\{ b_i\ :\ i\in [k] \big\} \\
        C := \big\{ c_i\ :\ i\in [k] \big\}\quad ;\quad D := \big\{ d_i\ :\ i\in [k] \big\}
\end{aligned}
\end{equation}

For each $i\in [k]$ we add the edges (shown in~\autoref{fig:main} using \textcolor[rgb]{1.00,0.00,1.00}{magenta} color)
%\begin{itemize}
%  \item $\Source(A):= \big\{ a_i \to \w_{i,1}^{\ell,1}\ :\ \ell\in [N] \big\}$
%  \item $\Sink(B):= \big\{ \w_{i,N}^{\ell,N} \to b_{i}\ :\ \ell\in [N] \big\}$
%\end{itemize}
\begin{equation}\label{eqn:source-sink-A-B}
\begin{aligned}
        \Source(A) := \big\{ a_i \to \w_{i,1}^{\ell,1}\ :\ \ell\in [N] \big\}\ ;\
        \Sink(B) := \big\{ \w_{i,N}^{\ell,N} \to b_{i}\ :\ \ell\in [N] \big\}
\end{aligned}
\end{equation}
For each $j\in [k]$ we add the edges (shown in~\autoref{fig:main} using \textcolor[rgb]{1.00,0.00,1.00}{magenta} color)
%\begin{itemize}
%  \item $\Source(C):= \big\{ c_j \to \w_{1,j}^{1,\ell}\ :\ \ell\in [N] \big\}$
%  \item $\Sink(D):= \big\{ \w_{N,j}^{N,\ell} \to d_{j}\ :\ \ell\in [N] \big\}$
%\end{itemize}
\begin{equation}\label{eqn:source-sink-C-D}
\begin{aligned}
        \Source(C) := \big\{ c_j \to \w_{1,j}^{1,\ell}\ :\ \ell\in [N] \big\}\ ;\
        \Sink(D) := \big\{ \w_{N,j}^{N,\ell} \to d_{j}\ :\ \ell\in [N] \big\}
\end{aligned}
\end{equation}
\end{enumerate}
%\end{bracketenumerate}

This completes the construction of the graph $G_1$ (see ~\autoref{fig:main}).
%Formally, we have
%\begin{equation}\label{eqn:defn-of-G1}
%\begin{aligned}
%        V(G_1) &:= A\cup B\cup C\cup D \cup \Big( \bigcup_{\substack{i\in [k]\\j\in [k]}} V(G_{i,j}) \Big) \cup \Big( \bigcup_{\substack{i\in [k-1]\\j\in [k]}} H_{i,j}^{i+1,j}\Big) \cup \Big( \bigcup_{\substack{i\in [k]\\j\in [k-1]}} V_{i,j}^{i,j+1}\Big) \\
%        E(G_1) &:= \Source(A)\cup \Sink(B) \cup \Source(C)\cup \Sink(D) \cup \Big( \bigcup_{\substack{i\in [k]\\j\in [k]}} E(G_{i,j}) \Big) \cup \\
%        &\Big( \bigcup_{\substack{i\in [k-1]\\j\in [k]}} \Path(H_{i,j}^{i+1,j})\cup \Matching\big( G_{i,j}, H_{i,j}^{i+1,j} \big) \cup \Matching\big( H_{i,j}^{i+1,j},G_{i+1,j} \big) \Big)
%\end{aligned}
%\end{equation}

\begin{claim}
\label{claim:G1-is-planar-and-dag}
\normalfont $G_1$ is a planar DAG
\end{claim}
\begin{proof}
~\autoref{fig:main} gives a planar embedding of $G_1$. It is easy to verify from the construction of $G_1$ described at the start of~\autoref{subsubsec:construction-of-G1} (see also~\autoref{fig:main}) that $G_1$ is a DAG.
%We now show that $G_1$ is a DAG by demonstrating a topological ordering. First observe that the following sequence is a topological order for the grid $G_{i,j}$ for any $1\leq i,j\leq k$: start with $\ell=1$ and list the vertices (in the order which they appear) of the path $\Row_{\ell}(G_{i,j})$ as $\ell$ increases to $N$. Denote this topological ordering for $G_{i,j}$ by $\Phi(G_{i,j})$. We now give a topological ordering of $G_1$ (which is left to the reader to verify):
%\begin{itemize}
%  \item $\{a_1, a_2, \ldots, a_k\}$ followed by $\{a_1, a_2, \ldots, a_k\}$
%
%  \item For each $1\leq j\leq k$
%        \begin{itemize}
%          \item For each $1\leq i\leq k-1$
%                \begin{itemize}
%                  \item
%                  \item
%                  \item
%                \end{itemize}
%        \end{itemize}
%
%  \item $\{b_1, b_2, \ldots, b_k\}$ followed by $\{d_1, d_2, \ldots, d_k\}$
%\end{itemize}
%%    \begin{itemize}
%%      \item Ro
%%      \item
%%    \end{itemize}
%%It is easy to see that each grid $G_{i,j}$ is a DAG
%
%
%\todo{Add a few sentences giving topological ordering saying why it is a DAG}
%\qed
\end{proof}

\subsubsection{Obtaining the graph $G_2$ from $G_1$ via the splitting operation}
\label{subsubsec:construction-of-G2}
%\textbf{Splitting} operation: \\

Observe (see~\autoref{fig:main}) that every (black) grid vertex in $G_1$ has in-degree two and out-degree two. Moreover, the two in-neighbors and two out-neighbors do not appear alternately. For each (black) grid vertex $z\in G_1$ we set up the notation:
\begin{definition}
\label{defn:lefty-right-topy-bottomy}
\normalfont
\textbf{(four neighbors of each grid vertex in $G_1$)} For each (black) grid vertex $\z\in G_1$ we define the following four vertices
\begin{itemize}
  \item $\lefty(\z)$ is the vertex to the left of $\z$ (as seen by the reader) which has an edge incoming into $\z$
  \item $\bottomy(\z)$ is the vertex below $\z$ (as seen by the reader) which has an edge incoming into $\z$
  \item $\righty(\z)$ is the vertex to the right of $\z$ (as seen by the reader) which has an edge outgoing from $\z$
  \item $\topy(\z)$ is the vertex above $\z$ (as seen by the reader) which has an edge outgoing from $\z$
\end{itemize}
\end{definition}

%\begin{equation}\label{eqn:lefty-right-topy-bottomy}
%\begin{aligned}
%        \lefty(z) &\text{is the vertex to the left (as seen by the reader) which has an edge incoming into} z \\
%        B &:= \big\{ b_i\ :\ i\in [k] \big\} \\
%        C &:= \big\{ c_i\ :\ i\in [k] \big\} \\
%        D &:= \big\{ d_i\ :\ i\in [k] \big\}
%\end{aligned}
%\end{equation}

We now define the splitting operation which allows us to obtain the graph $G_2$ from the graph $G_1$ constructed in~\autoref{subsubsec:construction-of-G1}.

\begin{definition}
\normalfont
\textbf{(splitting operation)} For each $i,j\in [k]$ and each $q,\ell\in [N]$
\begin{itemize}
    \item If $(q,\ell)\notin S_{i,j}$, then we \textbf{split} the vertex $\w_{i,j}^{q,\ell}$ into two \textbf{distinct} vertices $\w_{i,j,\LB}^{q,\ell}$ and $\w_{i,j,\TR}^{q,\ell}$ and add the edge $\w_{i,j,\LB}^{q,\ell}\to \w_{i,j,\TR}^{q,\ell}$ (denoted by the dotted edge in~\autoref{fig:split}). The 4 edges (see~\autoref{defn:lefty-right-topy-bottomy}) incident on $\w_{i,j}^{q,\ell}$ are now changed as follows (see \autoref{fig:split}):
                \begin{itemize}
                \item Replace the edge $\lefty(\w_{i,j}^{q,\ell})\to \w_{i,j}^{q,\ell}$ by the edge $\lefty(\w_{i,j}^{q,\ell})\to \w_{i,j,\LB}^{q,\ell}$
                \item Replace the edge $\bottomy(\w_{i,j}^{q,\ell})\to \w_{i,j}^{q,\ell}$ by the edge $\bottomy(\w_{i,j}^{q,\ell})\to \w_{i,j,\LB}^{q,\ell}$
                \item Replace the edge $\w_{i,j}^{q,\ell}\to \righty(\w_{i,j}^{q,\ell})$ by the edge $\w_{i,j,\TR}^{q,\ell}\to \righty(\w_{i,j}^{q,\ell})$
                \item Replace the edge $\w_{i,j}^{q,\ell}\to \topy(\w_{i,j}^{q,\ell})$ by the edge $\w_{i,j,\TR}^{q,\ell}\to \topy(\w_{i,j}^{q,\ell})$
                %\item Add the edge $\w_{i,j,\LB}^{q,\ell}\to \w_{i,j,\TR}^{q,\ell}$ (denoted by the dotted edge in~\autoref{fig:split}).
                \end{itemize}
    \item Otherwise, if $(q,\ell)\in S_{i,j}$ then the vertex $\w_{i,j}^{q,\ell}$ is \textbf{not split}, and we define $\w_{i,j,\LB}^{q,\ell}=\w_{i,j}^{q,\ell}=\w_{i,j,\TR}^{q,\ell}$.
        %In this case, we add a dotted edge $\w_{i,j,\LB}^{q,\ell}\to \w_{i,j,\TR}^{q,\ell}$ which is a self-loop.
        Note that the four edges (\autoref{defn:lefty-right-topy-bottomy}) incident on $\w_{i,j}^{q,\ell}$ are unchanged.
\end{itemize}
\label{defn:splitting-operation}
\end{definition}

\begin{remark}
\label{remark:issue-of-self-loops-for-splitting}
\normalfont
To avoid case distinctions in the forthcoming proof of correctness of the reduction, we will use the following non-standard notation: the edge $s\leadsto s$ \textbf{does not} represent a self-loop but rather is to be viewed as \emph{``just staying put"} at the vertex $s$. Note that this does not affect edge-disjointness.
%Henceforth, we will use the notation $s\leadsto s$ not to denote a self-loop (as is standard) but rather
%%\todo{Needs polishing: basically every vertex is split, but some are split into two vertices and some are not}.
%
%Note that according to~\autoref{defn:splitting-operation} all black grid vertices of $G_1$ are split: however, some of them are split into two distinct vertices and some are split into the same vertex itself. In the first case, the dotted edge is an ``actual'' edge between two distinct vertices and in the second case the dotted edge is a self-loop. Note that self-loops are irrelevant when we consider edge-disjointness of paths since any self-loop which is part of a path can simply be removed from it. We use this slightly counterintuitive notation to avoid case distinctions in the forthcoming proof of correctness of the reduction.
%%Fix any $(i,j)\in [k]\times [k]$ and $(q,\ell)\in [N]\times [N]$.
%%
%%If  $\x_{\LB}\to \x_{\TR}$ is a self-loop (which is not an issue for edge-disjointness), but if $\x$ is split then $\x_{\LB}\to \x_{\TR}$ is an actual edge in $G_2$ and potentially might be used on multiple paths.
%%in this definition when we describe an edge $\x\to \x$ then this is not be considered as a self-loop but rather just the vertex $\x$.
%%We use this notation to avoid case distinctions whether the vertex $\x$ is split or not.
\end{remark}

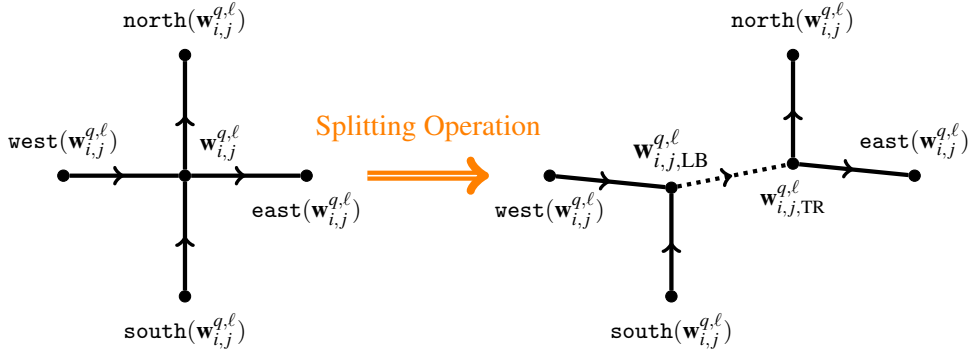
\begin{figure}[t!]
\centering
\begin{tikzpicture}[
vertex/.style={circle, draw=black, fill=black, text width=1.5mm, inner sep=0pt},
%every path/.style={ultra thick,middlearrow={>}},
scale=0.8]
\node[vertex, label=above right:\footnotesize{$\w_{i,j}^{q,\ell}$}] (v) at (0,0) {} ;
\node[vertex, label=above:\footnotesize{$\lefty(\w_{i,j}^{q,\ell})$}] (l) at (-2,0) {};
\node[vertex, label=below:\footnotesize{$\righty(\w_{i,j}^{q,\ell})$}] (r) at (2,0) {};
\node[vertex, label=below:\footnotesize{$\bottomy(\w_{i,j}^{q,\ell})$}] (b) at (0,-2) {};
\node[vertex, label=above:\footnotesize{$\topy(\w_{i,j}^{q,\ell})$}] (t) at (0,2) {};
\draw[ultra thick, middlearrow={>}] (v) -- (t);
\draw[ultra thick, middlearrow={>}] (v) -- (r);
\draw[ultra thick, middlearrow={>}] (l) -- (v);
\draw[ultra thick, middlearrow={>}] (b) -- (v);

\draw[orange,double, ultra thick,->] (3,0) -- node[above=3mm, draw=none, fill=none, rectangle] {Splitting Operation} (5,0);

\node[vertex, label=below:\footnotesize{$\w_{i,j,\TR}^{q,\ell}$}] (vtr) at (10,0.2) {} ;
\node[vertex, label=above:$\w_{i,j,\LB}^{q,\ell}$] (vlb) at (8,-0.2) {} ;
\node[vertex, label=below:\footnotesize{$\lefty(\w_{i,j}^{q,\ell})$}] (l) at (6,0) {};
\node[vertex, label=above:\footnotesize{$\righty(\w_{i,j}^{q,\ell})$}] (r) at (12,0) {};
\node[vertex, label=below:\footnotesize{$\bottomy(\w_{i,j}^{q,\ell})$}] (b) at (8,-2) {};
\node[vertex, label=above:\footnotesize{$\topy(\w_{i,j}^{q,\ell})$}] (t) at (10,2) {};
\draw[ultra thick, middlearrow={>}] (vtr) -- (t);
\draw[ultra thick, middlearrow={>}] (vtr) -- (r);
\draw[ultra thick, middlearrow={>}] (l) -- (vlb);
\draw[ultra thick, middlearrow={>}] (b) -- (vlb);
%\draw[black,orange,middlearrow={>}] (vlb)--(vtr);
\draw[dotted,ultra thick,middlearrow={>}] (vlb) -- (vtr);
\end{tikzpicture}

\caption{The splitting operation for the vertex $\w_{i,j}^{q,\ell}$ when
$(q,\ell)\notin S_{i,j}$. The idea behind this splitting is if we want edge-disjoint paths then we can go \textbf{either} left-to-right or bottom-to-top but not in \textbf{both} directions. On the other hand, if $(q,\ell)\in S_{i,j}$ then the picture on the right-hand side (after the splitting operation) would look exactly like that on the left-hand side.
%\todo{Needs polishing}
%However, if we just want to go from left to right (top to bottom) then it is possible by orienting the dotted edge to the right (left), respectively.
\label{fig:split}
}
\end{figure}

We are now ready to define the graph $G_2$ and the set $\mathcal{T}$ of terminal pairs:

\begin{definition}
\normalfont The graph $G_2$ is obtained by applying the splitting operation (~\autoref{defn:splitting-operation}) to each (black) grid vertex of $G_1$, i.e., the set of vertices given by $\bigcup_{1\leq i,j\leq k} V(G_{i,j})$. The set of terminal pairs is $\mathcal{T}:= \big\{(a_i, b_i) : i\in [k] \big\}\cup \big\{(c_j, d_j) : j\in [k] \big\}$
\label{defn:G2}
\end{definition}

%This concludes the construction of the graph $G_2$: it is obtained from $G_1$ by applying the splitting operation (\autoref{defn:splitting-operation}) to each (black) grid vertex.

Note that in $G_2$ we have
\begin{itemize}
  \item All vertices in $G_{2}$ except $A\cup C$ have out-degree at most $2$
  \item All vertices in $G_{2}$ except $B\cup D$ have in-degree at most $2$
\end{itemize}
We will later show (see last paragraph in the proof of~\autoref{thm:main-result}) how to edit $G_2$ such that each vertex has both in-degree and out-degree at most $2$. The next claim shows that $G_2$ is also both planar and acyclic (like $G_1$).

\begin{claim}
\label{claim:G2-is-planar-and-dag}
\normalfont
$G_2$ is a planar DAG
\end{claim}
\begin{proof}
In~\autoref{claim:G1-is-planar-and-dag}, we have shown that $G_1$ is a planar DAG. By~\autoref{defn:G2}, $G_2$ is obtained from $G_1$ by applying the splitting operation (\autoref{defn:splitting-operation}) on every (black) grid vertex, i.e., every vertex from the set $\bigcup_{1\leq i,j\leq k} V(G_{i,j})$.

By~\autoref{defn:lefty-right-topy-bottomy}, every vertex of $G_1$ that is split  has exactly two in-neighbors and two out-neighbors in $G_1$. Hence, it is easy to see (\autoref{fig:split}) that the splitting operation (\autoref{defn:splitting-operation}) does not destroy planarity when we construct $G_2$ from $G_1$. Since $G_1$ is a DAG, replacing each split (black) grid vertex $\w$ in $G_1$ by $\w_{\LB}$ followed by $\w_{\TR}$ in the topological order of $G_1$ gives a topological order for $G_2$. Hence, $G_2$ is a planar DAG.
%\todo{Add a few sentences giving topological ordering saying why it is a DAG}
%\qed
\end{proof}

We now set up notation for the grids in $G_2$:

\begin{definition}
\label{defn:G-ij-split}
\normalfont For each $i,j \in [k]$, we define $G_{i,j}^{\splitt}$ to be the graph obtained by applying the splitting operation (\autoref{defn:splitting-operation}) to each vertex of $G_{i,j}$. For each $i,j\in [k]$ and each $q,\ell\in [N]$ we define $\splitt(\w_{i,j}^{q,\ell}):= \big\{ \w_{i,j,\LB}^{q,\ell}, \w_{i,j,\TR}^{q,\ell} \big\}$.
%\begin{align*}
% \splitt(\w_{i,j}^{q,\ell}) &:=  \big\{ \w_{i,j}^{q,\ell} \big\}\ \text{if } \w_{i,j}^{q,\ell}\in S_{i,j} \\
% \splitt(\w_{i,j}^{q,\ell}) &:= \big\{ \w_{i,j,\LB}^{q,\ell}, \w_{i,j,\TR}^{q,\ell} \big\}\ \text{if } \w_{i,j}^{q,\ell}\notin S_{i,j}
%\end{align*}
%for each $q,\ell \in [N]$
\end{definition}

%Note that for every (black) grid vertex $\splitt(w)$ contains two vertices if $w$ is split. Otherwise $\splitt(w)$ just contains one vertex since in this case $\w_{\LB}=\w=\w_{\TR}$ by~\autoref{defn:splitting-operation}.

\subsection{ Solution for \edp $\Rightarrow$ Solution for \gtleq
%Instance $(G_2,\mathcal{T})$ of \edp has a solution $\Rightarrow$ instance $(k,N,\mathcal{S})$ of \gtleq has a solution
}
\label{subsec:hard-direction}

In this section, we show that if the instance $(G_2, \mathcal{T})$ of \edp has a solution then the instance $(k, N, \mathcal{S})$ of \gtleq also has a solution.

Suppose that the instance $(G_2,\mathcal{T})$ of \edp has a solution, i.e., there is a collection of $2k$ pairwise edge-disjoint paths $\big\{ P_1, P_2, \ldots, P_k,$ $Q_1, Q_2, \ldots, Q_k \big\}$ in $G_2$ such that
\begin{equation}\label{eqn:paths-Pi-Qj-disjoint}
\begin{aligned}
        &P_{i}\ \text{is an}\ a_i \leadsto b_i\ \text{path}\ \forall\ i\in [k]\\ %\quad \text{and}\quad
        &Q_{j}\ \text{is an}\ c_j \leadsto d_j\ \text{path}\ \forall\ j\in [k]
\end{aligned}
\end{equation}
%
%\begin{equation}\label{eqn:paths-Pi-Qj-disjoint}
%\begin{aligned}
%        P_{i} &\ \text{is an}\ a_i \leadsto b_i\ \text{path for each}\ i\in [k]\\
%        Q_{j} &\ \text{is an}\ c_j \leadsto d_j\ \text{path for each}\ j\in [k]
%\end{aligned}
%\end{equation}

\noindent To streamline the arguments of this section, we define the following subsets of vertices of $G_2$:

\begin{definition}
\label{defn:horizontal-vertical-sets}
%\normalfont
\textbf{(horizontal \& vertical levels)}\\
For each $j\in [k]$, we define the following set of vertices:
$$ \Horizontal(j) = \{ c_j, d_j \} \cup \Big( \bigcup_{i=1}^{k} V(G_{i,j}^{\splitt})\Big) \cup \Big( \bigcup_{i=1}^{k-1} H_{i,j}^{i+1,j} \Big) $$
For each $i\in [k]$, we define the following set of vertices:
$$ \Vertical(i) = \{ a_i, b_i \} \cup \Big( \bigcup_{j=1}^{k} V(G_{i,j}^{\splitt})\Big) \cup \Big( \bigcup_{j=1}^{k-1} V_{i,j}^{i,j+1} \Big) $$
\end{definition}

From~\autoref{defn:horizontal-vertical-sets}, it is easy to verify that $\Vertical(i)\cap \Vertical(i') = \emptyset = \Horizontal(i)\cap \Horizontal(i')$ for every $1\leq i\neq i'\leq k$.

\begin{definition}
%\normalfont
\textbf{(boundary vertices in $G_2$)}
For each $1\leq i,j\leq k$ we define the set of boundary vertices of the grid $G_{i,j}^{\splitt}$ in the graph $G_2$ as follows:
\begin{equation}\label{eqn:left-right-top-bottom-G1}
 %\left.
 \begin{aligned}
        \Le(G_{i,j}^{\splitt}) := \big\{ \w_{i,j,\LB}^{1,\ell}\ :\ \ell\in [N]  \big\}\ ;\
        \Ri(G_{i,j}^{\splitt}) := \big\{ \w_{i,j,\TR}^{N,\ell}\ :\ \ell\in [N]  \big\} \\
        \To(G_{i,j}^{\splitt}) := \big\{ \w_{i,j,\TR}^{\ell,N}\ :\ \ell\in [N]  \big\}\ ;\
        \Bo(G_{i,j}^{\splitt}) := \big\{ \w_{i,j,\LB}^{\ell,1}\ :\ \ell\in [N]  \big\}
       \end{aligned}
\end{equation}
\label{defn:left-right-top-bottom-sets-in-G2}
\end{definition}

\begin{lemma}%$[\star]$\footnote{Proofs of all results labeled with $\star$ are deferred to the Appendix (\autoref{secapp:missing-proofs})}
\normalfont
For each $i\in [k]$ the path $P_i$ satisfies the following two structural properties:
\begin{itemize}
  \item every edge of the path $P_i$ has both end-points in $\Vertical(i)$
  \item $P_i$ contains an $\Bo(G_{i,j}^{\splitt}) \leadsto \To(G_{i,j}^{\splitt})$ path for each $j\in [k]$.
\end{itemize}
\label{lem:Pi-structural-characterizations}
\end{lemma}
\begin{proof}
%\todo{To be added}
For this proof, define $H_{0,j}^{1,j}:= \{c_j\}$ and $H_{k,j}^{k+1,j}:= \{d_j\}$ for each $j\in [k]$.

Fix any $i^*\in [k]$. Note that $P_{i^*}$ is an $a_{i^*}\leadsto b_{i^*}$ path and hence starts and ends at a vertex in $\Vertical(i^*)$. We now prove the first part of lemma by showing two claims which state that $P_{i^*}$ cannot contain any vertex of $N_{G_2}^{+}\big( \Vertical(i^*) \big)$ and $N_{G_2}^{-}\big( \Vertical(i^*) \big)$ respectively.
    \begin{claim}
    \normalfont
    $P_{i^*}$ does not contain any vertex of $N_{G_2}^{+}\big( \Vertical(i^*) \big)$.
    \end{claim}
    \begin{proof}
     The structure of $G_2$ implies that
        \begin{itemize}
          \item $N_{G_2}^{+}\big( \Vertical(i) \big) = \bigcup_{j=1}^{k} H_{i,j}^{i+1,j}$ for each $i\in [k]$
          \item $N_{G_2}^{+}\big( \bigcup_{j=1}^{k} H_{i,j}^{i+1,j} \big) \subseteq \Vertical(i+1)$ for each $0\leq i\leq k-1$
          \item $N_{G_2}^{+}\big( \bigcup_{j=1}^{k} H_{k,j}^{k+1,j} \big) =\emptyset$ since each vertex of $D$ is a sink in $G_2$
        \end{itemize}
         Hence, if $P_{i^*}$ contains a vertex from $N_{G_2}^{+}\big( \Vertical(i^*) \big)$ then it cannot ever return back to $\Vertical(i^*)$ which contradicts the fact that the last vertex of $P_{i^*}$ is $b_{i^*}\in \Vertical(i^*)$.
    %\qed
    \end{proof}

    \begin{claim}
    \normalfont
    $P_{i^*}$ does not contain any vertex of $N_{G_2}^{-}\big( \Vertical(i^*) \big)$.
    \end{claim}
    \begin{proof}
    The structure of $G_2$ implies that
        \begin{itemize}
          \item $N_{G_2}^{-}\big( \Vertical(i) \big) = \bigcup_{j=1}^{k} H_{i-1,j}^{i,j}$ for each $i\in [k]$
          \item $N_{G_2}^{-}\big( \bigcup_{j=1}^{k} H_{i,j}^{i+1,j} \big) \subseteq \Vertical(i)$ for each $1\leq i\leq k$
          \item $N_{G_2}^{-}\big( \bigcup_{j=1}^{k} H_{0,j}^{1,j} \big) =\emptyset$ since each vertex of $C$ is a source in $G_2$
        \end{itemize}
        Hence, if $P_{i^*}$ contains a vertex from $N_{G_2}^{-}\big( \Vertical(i^*) \big)$ then $P_{i^*}$ cannot have started at a vertex of $\Vertical(i^*)$ which contradicts the fact that the first vertex of $P_{i^*}$ is $a_{i^*}\in \Vertical(i^*)$.
%         Since the first vertex of $P_{i^*}$ is $a_{i^*}\in \Vertical(i^*)$, it follows that $P_{i^*}$ cannot contain a vertex of $N_{G_2}^{-}\big( \Vertical(i^*) \big)$.
         %ever reaches a vertex of $N_{G_2}^{-}\big( \Vertical(i^*) \big)$ then it cannot ever return back to $\Vertical(i^*)$ which contradicts the fact that the last vertex of $P_{i^*}$ is $b_{i^*}\in \Vertical(i^*)$.
    %\qed
    \end{proof}
    This concludes the proof of the first part of the lemma. We now show the second part of the lemma. We define $V_{i^*,0}^{i^*,1}:=\{a_{i^*}\}$ and $V_{i^*,k}^{i^*,k+1}:=\{b_{i^*}\}$.
    The structure of $G_2$ implies that
        \begin{itemize}
        \item $N^{+}_{G_{2}[\Vertical(i^*)]}\big( G_{i^*,j}^{\splitt}\big) = V_{i^*,j}^{i^*,j+1}$ and $N^{-}_{G_{2}[\Vertical(i^*)]}\big( G_{i^*,j}^{\splitt}\big) = V_{i^*,j-1}^{i^*,j}$ for each $j\in [k]$

        \item $N^{+}_{G_{2}[\Vertical(i^*)]}\big( V_{i^*,j}^{i^*,j+1} \big) = \Bo\big( G_{i^*,j+1}^{\splitt} \big)$ for each $0\leq j\leq k-1$

        \item $N^{-}_{G_{2}[\Vertical(i^*)]}\big( V_{i^*,j}^{i^*,j+1} \big) = \To\big( G_{i^*,j}^{\splitt} \big)$ for each $1\leq j\leq k$
        \end{itemize}
        These three relations, combined with the first part of the lemma which states that $P_{i*}$ lies within $G_{2}[\Vertical(i^*)]$, implies that $P_{i^*}$ contains an $\Bo(G_{i^*,j}^{\splitt}) \leadsto \To(G_{i^*,j}^{\splitt})$ path for each $j\in [k]$.
%~\\
%The structure of $G_2$ implies that for each $i\in [k]$
%\begin{itemize}
%  \item $N_{G_2}^{+}\big( \Vertical(i) \big) = \bigcup_{j=1}^{k-1} H_{i,j}^{i+1,j}$ and $N_{G_2}^{-}\big( \Vertical(i) \big) = \bigcup_{j=1}^{k-1} H_{i-1,j}^{i,j}$
%  \item
%\end{itemize}
%
%\qed
This concludes the proof of~\autoref{lem:Pi-structural-characterizations}.
\end{proof}

The proof of the next lemma is very similar to that of~\autoref{lem:Pi-structural-characterizations}, and we skip repeating the details.

\begin{lemma}
\normalfont
For each $j\in [k]$ the path $Q_j$ satisfies the following two structural properties:
\begin{itemize}
  \item every edge of the path $Q_j$ has both end-points in $\Horizontal(j)$
  \item $Q_j$ contains an $\Le(G_{i,j}^{\splitt}) \leadsto \Ri(G_{i,j}^{\splitt})$ path for each $i\in [k]$
\end{itemize}
\label{lem:Qj-structural-characterizations}
\end{lemma}
%\begin{proof}
%%\todo{To be added, but very similar to that of~\autoref{lem:Pi-structural-characterizations}.}
%The proof of this lemma is very similar to that of~\autoref{lem:Pi-structural-characterizations}, and we omit the details here.
%\qed
%\end{proof}
%%
%
%\newpage
%
\begin{lemma}%$[\star]$
\normalfont
For any $(i,j)\in [k]\times [k]$, let $P', Q'$ be any $\Bo(G_{i,j}^{\splitt}) \leadsto \To(G_{i,j}^{\splitt})$, $\Le(G_{i,j}^{\splitt}) \leadsto \Ri(G_{i,j}^{\splitt})$ paths in $G_2$ respectively. If $P'$ and $Q'$ are edge-disjoint then there exists $(\mu,\delta)\in S_{i,j}$ such that the vertex $\w_{i,j,\LB}^{\mu,\delta}=\w_{i,j}^{\mu,\delta}=\w_{i,j,\TR}^{\mu,\delta}=$ belongs to both $P'$ and $Q'$
%such that $P'$ and $Q$' then $P'$
\label{lem:any-left-to-right-path-must-intersect-bottom-to-top-path}
\end{lemma}
\begin{proof}
%\todo{To be added!}
%
Let $P'',Q''$ be the paths obtained from $P',Q'$ by contracting all the dotted edges on $P',Q'$ respectively. By the construction of $G_2$ (\autoref{defn:G2}) and the splitting operation (\autoref{defn:splitting-operation}), it follows that $P'',Q''$ are $\Bo(G_{i,j}) \leadsto \To(G_{i,j}), \Le(G_{i,j}) \leadsto \Ri(G_{i,j})$ paths in $G_1$ respectively. Hence, there exist $x_1,x_2 \in [N]$ such that $P''$ is a $\w_{i,j}^{x_1,1}\to \w_{i,j}^{x_2,N}$ path and $y_1,y_2 \in [N]$ such that $Q''$ is a $\w_{i,j}^{1,y_1}\to \w_{i,j}^{N,y_2}$ path. We now show that $P''$ and $Q''$ must intersect in $G_1$
        \begin{claim}
        %\normalfont
        $P''$ and $Q''$ have a common vertex in $G_1$
        \label{claim:P"-Q"-intersect-in-G1}
        \end{claim}
        \begin{proof}
        For each $x\in [N]$ such that $x_1\leq x\leq x_2$ define $P''(x)=\big\{ y\in [N] : \w_{i,j}^{x,y}\in P'' \big\}$. For each $x\in [N]$ such that $x_1\leq x\leq x_2$ define $Q''(x)=\big\{ y\in [N] : \w_{i,j}^{x,y}\in Q'' \big\}$. We will prove the claim by showing that there exists $x^*,y^*\in [N]$ such that $y^*\in \big(P''(x^*)\cap Q''(x^*)\big)$. By the orientation of the edges in $G_{i,j}$, it follows that
            \begin{equation}\label{eqn:left-right-top-bottom-G1}
 %\left.
                \begin{aligned}
                &\max\ P''(z) = \min\ P''(z+1)\ \text{and } \max\ Q''(z) = \min\ Q''(z+1)\quad \forall\ x_1\leq z< x_2 \\
                &\text{If } 1\leq u\leq z\leq N\  \text{then }  \max P''(u) \leq \min\ P''(z)\ \text{and } \max\ Q''(u) \leq \min Q''(z)
                \end{aligned}
        \end{equation}

       By definition of $Q''$, we have $y_1 \in Q''(1)$ and hence $y\geq y_1 \geq 1$ for each $y\in Q''(x_1)$. If $\big(P''(x_1)\cap Q''(x_1)\big)\neq \emptyset$ then we are done. Otherwise, we have that $\min\ Q''(x_1) > \max\ P''(x_1)$ since $1\in P''(x_1)$. Now if $\big(P''(x_1 +1)\cap Q''(x_1 +1 )\big)\neq \emptyset$ then we are done. Otherwise, we have $\min Q''(x_1 + 1) > \max P''(x_1+1)$ since $\min Q''(x_1 + 1) = \max Q''(x_1)$. Continuing this way, we must find an $x^*\in \mathbb{N}$ such that $x_1\leq x^*\leq x_2$ and $\big(P''(x^*)\cap Q''(x^*)\big)\neq \emptyset$: this is because $N\in P''(x_2)$ and hence $\min Q''(x_2)\leq N = \max P''(x_2)$. Since $\big(P''(x^*)\cap Q''(x^*)\big)\neq \emptyset$ let $y^*\in \big(P''(x^*)\cap Q''(x^*)\big)$, i.e., the vertex $\w_{i,j}^{x^*, y^*}$ belongs to both $P''$ and $Q''$.
%        Similarly, we have $y_2 \in Q''(N)$ and hence $y\leq y_2 \leq N$ for each $y\in Q''(x_2)$.
%
%        Note that $1\in P''(x_1)$ and $N\in P''(x_2)$. By definition of $Q''$, we have $y_1 \in Q''(1)$ and hence $y\geq y_1 \geq 1$ for each $y\in Q''(x_1)$. Similarly, we have $y_2 \in Q''(N)$ and hence $y\leq y_2 \leq N$ for each $y\in Q''(x_2)$.
        %
        %We view $P''$ and $Q''$ as continuous functions. For each $x\in [x_1,x_2]$ define $f_{P''}(x)=y$ (resp. $f_{Q''}(x)=y$) if $\w_{i,j}^{x,y}\in P''$ (resp. $\w_{i,j}^{x,y}\in Q''$). By the orientations of the edges (\autoref{fig:main}) in $G_1$, it follows that both the functions $f_{P''}$ and $f_{Q''}$ are continuous and non-decreasing on the domain $[x_1,x_2]$. Hence, we have $f_{P''}(x_1)=1\leq y_1 = f_{Q''}(1) \leq f_{Q''}(x_1)$ and $f_{P''}(x_2)=N\geq y_2 = f_{Q''}(N)\geq f_{Q''}(x_2)$. Therefore, by the intermediate value theorem, there exists $x\in [x_1, x_2]$ such that $f_{P''}(x)=f_{Q''}(x)$, i.e.,
        %\qed
        \end{proof}

        By~\autoref{claim:P"-Q"-intersect-in-G1}, the paths $P'',Q''$ have a common vertex in $G_1$. Let this vertex be $\w_{i,j}^{\mu,\delta}$. Viewing the paths $P'', Q''$ in $G_2$, i.e., ``un-contracting" the dotted edges (\autoref{defn:splitting-operation}), it follows that both $P'$ and $Q'$ share the dotted edge $\w_{i,j}^{\mu,\delta,\LB}\to \w_{i,j,\TR}^{\mu,\delta}$. Since $P'$ and $Q'$ are given to be edge-disjoint, this implies that the edge $\w_{i,j}^{\mu,\delta,\LB}\to \w_{i,j,\TR}^{\mu,\delta}$ cannot exist in $G_2$, i.e., $(\mu,\delta)\in S_{i,j}$ and the vertex $\w_{i,j,\LB}^{\mu,\delta}=\w_{i,j}^{\mu,\delta}=\w_{i,j,\TR}^{\mu,\delta}$ belongs to both $P'$ and $Q'$ (recall~\autoref{defn:splitting-operation}).
        %\qed
        This concludes the proof of~\autoref{lem:any-left-to-right-path-must-intersect-bottom-to-top-path}.
\end{proof}

\begin{lemma}%$[\star]$
\normalfont
The instance $(k, N, \mathcal{S})$ of \gtleq has a solution.
\label{lem:gtleq-has-a-soln}
\end{lemma}
\begin{proof}
%\todo{To be added}
%
Fix any $(i,j)\in [k]\times [k]$. By~\autoref{lem:Pi-structural-characterizations}, $P_i$ contains an $\Bo(G_{i,j}^{\splitt}) \leadsto \To(G_{i,j}^{\splitt})$ path say $P_{i,j}$. By~\autoref{lem:Qj-structural-characterizations}, $Q_j$ contains an $\Le(G_{i,j}^{\splitt}) \leadsto \Ri(G_{i,j}^{\splitt})$ path say $Q_{i,j}$. Since $P_i$ and $Q_j$ are edge-disjoint (\autoref{eqn:paths-Pi-Qj-disjoint}), it follows that the paths $P_{i,j}$ and $Q_{i,j}$ are also edge-disjoint. Applying~\autoref{lem:any-left-to-right-path-must-intersect-bottom-to-top-path} to the paths $P_{i,j}$ and $Q_{i,j}$ we get that there exists $(\mu_{i,j},\delta_{i,j})\in [N]\times [N]$ such that $(\mu_{i,j},\delta_{i,j})\in S_{i,j}$ and the vertex $\w_{i,j,\LB}^{\mu_{i,j},\delta_{i,j}}=\w_{i,j}^{\mu_{i,j},\delta_{i,j}}=\w_{i,j,\TR}^{\mu_{i,j},\delta_{i,j}}$ belongs to $P_{i,j}$ (and hence also to $P_i$) and $Q_{i,j}$ (and hence also to $Q_j$).

We now claim that the values $\big\{ (\mu_{i,j},\delta_{i,j}) : (i,j)\in [k]\times[k] \big\}$ form a solution for the instance $(k,N,\mathcal{S})$ of \gtleq. In the last paragraph, we have already shown that $(\mu_{i,j},\delta_{i,j})\in S_{i,j}$ for each $(i,j)\in [k]\times[k]$. For each $(i,j)\in [k-1]\times [k]$ both the vertices $\w_{i,j,\LB}^{\mu_{i,j},\delta_{i,j}}=\w_{i,j,\TR}^{\mu_{i,j},\delta_{i,j}}$ and $\w_{i+1,j,\LB}^{\mu_{i+1,j},\delta_{i+1,j}}=\w_{i+1,j,\TR}^{\mu_{i+1,j},\delta_{i+1,j}}$ belong to the path $Q_{j}$ which is contained in $G_{2}[\Horizontal(j)]$ (\autoref{lem:Qj-structural-characterizations}). Hence, by the orientation of the edges in $G_2$, it follows that $\delta_{i,j}\leq \delta_{i+1,j}$. Similarly, it can be shown that $\mu_{i,j}\leq \mu_{i,j+1}$ for each $(i,j)\in [k]\times [k-1]$.
%This concludes the proof of the lemma.
%%Basic idea is that $P_i$ and $Q_j$ are edge-disjoint. So contracting all dotted edges from $G_2$ doesn't change that the corresponding paths are edge-disjoint, i.e., the collection of $2k$ paths $\{P'_1, P'_2, \ldots, P'_k, Q'_1, Q'_2, \ldots, Q'_k\}$ are pairwise edge-disjoint in $G_1$ where $P'_i, Q'_i$ are obtained from $P'_i, Q'_i$ respectively by contracting all dotted edges.
%\qed
\end{proof}

\subsection{Solution for \gtleq $\Rightarrow$ Solution for \edp
%Instance $(k,N,\mathcal{S})$ of \gtleq has a solution $\Rightarrow$ instance $(G_2,\mathcal{T})$ of \edp has a solution
}
\label{subsec:easy-direction}

In this section, we show that if the instance $(k, N, \mathcal{S})$ of \gtleq has a solution then the instance $(G_2, \mathcal{T})$ of \edp also has a solution.

Suppose that the instance $(k, N, \mathcal{S})$ of \gtleq has a solution given by the pairs $\big\{(\alpha_{i,j}, \beta_{i,j})\ : i,j\in [k]\big\}$. Hence, we have
%\todo{indices might have to be swapped}
\begin{equation}\label{eqn:alpha-beta-gtleq-inequalities}
\begin{aligned}
        \big( \alpha_{i,j}, \beta_{i,j} \big)\in S_{i,j} & \quad \text{for each }(i,j)\in [k]\times [k]\\
        \alpha_{i,j}\leq \alpha_{i,j+1} & \quad \text{for each }(i,j)\in [k]\times [k-1] \\
        \beta_{i,j}\leq \beta_{i+1,j} & \quad \text{for each }(i,j)\in [k-1]\times [k]
\end{aligned}
\end{equation}

\begin{definition}
\normalfont
\textbf{(row-paths and column-paths in $G_2$)} For each $(i,j)\in [k]\times [k]$ and $\ell\in [N]$ we define
\begin{itemize}
  \item $\Row_{\ell}(G_{i,j}^{\splitt})$ to be the $\w_{i,j,\LB}^{1,\ell}\leadsto \w_{i,j,\TR}^{N,\ell}$ path in $G_2[G_{i,j}^{\splitt}]$ consisting of the following edges (in order): for each $r\in [N-1]$
        \begin{itemize}
          \item $\w_{i,j,\LB}^{r,\ell}\to \w_{i,j,\TR}^{r,\ell}$ and $\w_{i,j,\TR}^{r,\ell}\to \w_{i,j,\LB}^{r+1,\ell}$
          %\item
        \end{itemize}
        followed finally by the edge $\w_{i,j,\LB}^{N,\ell}\to \w_{i,j,\TR}^{N,\ell}$

  \item $\Column_{\ell}(G_{i,j}^{\splitt})$ to be the $\w_{i,j,\LB}^{\ell,1}\leadsto \w_{i,j,\TR}^{\ell,N}$ path in $G_2$ consisting of the following edges (in order): for each $r\in [N-1]$
        \begin{itemize}
          \item $\w_{i,j,\LB}^{\ell,r}\to \w_{i,j,\TR}^{\ell,r}$ and $\w_{i,j,\TR}^{\ell,r}\to \w_{i,j,\LB}^{\ell,r+1}$
          %\item
        \end{itemize}
        followed finally by the edge $\w_{i,j,\LB}^{\ell,N}\to \w_{i,j,\TR}^{\ell,N}$
\end{itemize}
%Note that if $\x$ is not split then $\x_{\LB}\to \x_{\TR}$ is a self-loop (which is not an issue for edge-disjointness), but if $\x$ is split then $\x_{\LB}\to \x_{\TR}$ is an actual edge in $G_2$ and potentially might be used on multiple paths.
%%in this definition when we describe an edge $\x\to \x$ then this is not be considered as a self-loop but rather just the vertex $\x$.
%We use this notation to avoid case distinctions whether the vertex $\x$ is split or not.
%$\Row_{\ell}(G_{i,j}) := \w_{i,j}^{1,\ell}\to \w_{i,j}^{2,\ell}\to \w_{i,j}^{3,\ell} \to \ldots \to \w_{i,j}^{N-1,\ell} \to \w_{i,j}^{N,\ell} $\\
%$\Column_{\ell}(G_{i,j}) :=  \w_{i,j}^{\ell,1}\to \w_{i,j}^{\ell,2}\to \w_{i,j}^{\ell,3} \to \ldots \to \w_{i,j}^{\ell,N-1} \to \w_{i,j}^{\ell,N}$
\label{defn:Column-Row-G2}
\end{definition}
Using the special types of paths from~\autoref{defn:Column-Row-G2}, we can now show the following lemma:

\begin{lemma}%{blah}$[\star]$
\label{lem:Ri-Tj-form-a-solution-for-edp}
%\normalfont
The instance $(G_2, \mathcal{T})$ of \edp has a solution.
\end{lemma}
\begin{proof}
We build a collection of $2k$ paths $\mathcal{P}:=\big\{ R_1, R_2, \ldots, R_k, T_1, T_2, \ldots, T_k\big\}$ and show that it forms a solution for the instance $(G_2, \mathcal{T})$ of \edp. First, we describe this collection of paths below:
\begin{description}

\item[- Description of the set of paths $\{R_1, R_2, \ldots, R_k\}$]: \\
        For each $i\in [k]$, we build the path $R_i$ as follows:
\begin{itemize}
  \item Start with the edge $a_i \to \w_{i,1,\LB}^{\alpha_{i,1},1}$
  \item For each $j\in [k-1]$ use the $\w_{i,j,\LB}^{\alpha_{i,j},1} \leadsto \w_{i,j+1,\LB}^{\alpha_{i,j+1},1}$ path obtained by concatenating
            \begin{itemize}
              \item the $\w_{i,j,\LB}^{\alpha_{i,j},1}\leadsto \w_{i,j,\TR}^{\alpha_{i,j},N}$ path $\Column_{\alpha_{i,j}}(G_{i,j}^{\splitt})$ from~\autoref{defn:Column-Row-G2}
              \item the $\w_{i,j,\TR}^{\alpha_{i,j},N} \leadsto \w_{i,j+1,\LB}^{\alpha_{i,j+1},1}$ path $\w_{i,j,\TR}^{\alpha_{i,j},N} \to \ve_{i,j}^{i,j+1}(\alpha_{i,j}) \to \cdots\cdots \to \ve_{i,j}^{i,j+1}(\alpha_{i,j+1}) \to \w_{i,j+1,\LB}^{\alpha_{i,j+1},1} $ which exists since~\autoref{eqn:alpha-beta-gtleq-inequalities} implies $\alpha_{i,j}\leq \alpha_{i,j+1}$.
            \end{itemize}
  \item Now, we have reached the vertex $\w_{i,k,\LB}^{\alpha_{i,k},1}$. Use the $\w_{i,k,\LB}^{\alpha_{i,k},1} \leadsto \w_{i,k,\TR}^{\alpha_{i,k},N}$ path  \linebreak $\Column_{\alpha_{i,k}}(G_{i,k}^{\splitt})$ from~\autoref{defn:Column-Row-G2} to reach the vertex $\w_{i,k,\TR}^{\alpha_{i,k},N}$.
  \item Finally, use the edge $\w_{i,k,\TR}^{\alpha_{i,k},N} \to b_i$ to reach $b_i$.
\end{itemize}

\item[- Description of the set of paths $\{T_1, T_2, \ldots, T_k\}$]:\\
For each $j\in [k]$, we build the path $T_j$ as follows:
\begin{itemize}
  \item Start with the edge $c_j \to \w_{1,j,\LB}^{1,\beta_{1,j}}$
  \item For each $i\in [k-1]$ use the $\w_{i,j,\LB}^{1,\beta_{i,j}} \leadsto \w_{i+1,j,\LB}^{1,\beta_{i+1,j}}$ path obtained by concatenating
            \begin{itemize}
              \item the $\w_{i,j,\LB}^{1,\beta_{i,j}}\leadsto \w_{i,j,\TR}^{N,\beta_{i,j}}$ path $\Row_{\beta_{i,j}}(G_{i,j}^{\splitt})$ from~\autoref{defn:Column-Row-G2}
              \item the $\w_{i,j,\TR}^{N,\beta_{i,j}} \leadsto \w_{i+1,j,\LB}^{1,\beta_{i+1,j}}$ path $\w_{i,j,\TR}^{N,\beta_{i,j}} \to \h_{i,j}^{i+1,j}(\beta_{i,j}) \to \cdots\cdots \to \h_{i,j}^{i+1,j}(\beta_{i+1,j}) \to \w_{i+1,j,\LB}^{1,\beta_{i+1,j}} $ which exists since~\autoref{eqn:alpha-beta-gtleq-inequalities} implies $\beta_{i,j}\leq \beta_{i+1,j}$.
            \end{itemize}
  \item Now, we have reached the vertex $\w_{k,j,\LB}^{1,\beta_{k,j}}$. Use the $\w_{k,j,\LB}^{1,\beta_{k,j}} \leadsto \w_{k,j,\TR}^{N,\beta_{k,j}}$ path \linebreak $\Row_{\beta_{k,j}}(G_{k,j}^{\splitt})$ from~\autoref{defn:Column-Row-G2} to reach the vertex $\w_{k,j,\TR}^{N,\beta_{k,j}}$.
  \item Finally, use the edge $\w_{k,j,\TR}^{N,\beta_{k,j}} \to d_j$ to reach $d_j$.
\end{itemize}

\end{description}

By~\autoref{defn:horizontal-vertical-sets}, it follows that every edge of the path $R_i$ has both endpoints in $\Vertical(i)$ for every $i\in [k]$. Since $\Vertical(i) \cap \Vertical(i')=\emptyset$ for every $1\leq i\neq i'\neq k$, it follows that the collection of paths $\{R_1, R_2, \ldots, R_k\}$ are pairwise edge-disjoint.

By~\autoref{defn:horizontal-vertical-sets}, it follows that every edge of the path $T_j$ has both endpoints in $\Horizontal(j)$ for every $j\in [k]$. Since $\Horizontal(j) \cap \Horizontal(j')=\emptyset$ for every $1\leq j\neq j'\neq k$, it follows that the collection of paths $\{T_1, T_2, \ldots, T_k\}$ are pairwise edge-disjoint.

Fix any $(i,j)\in [k]\times [k]$. We now conclude the proof of this lemma by showing that $R_i$ and $T_j$ are edge-disjoint. By the construction of $G_2$ (\autoref{fig:main} and~\autoref{fig:split}) and definitions of the paths $R_i$ and $T_j$, it follows that the only common edge between $R_i$ and $T_j$ could be $\w_{i,j,\LB}^{\alpha_{i,j}, \beta_{i,j}}\to \w_{i,j,\TR}^{\alpha_{i,j}, \beta_{i,j}}$. By~\autoref{eqn:alpha-beta-gtleq-inequalities}, we have that $(\alpha_{i,j}, \beta_{i,j})\in S_{i,j}$. Hence, by the splitting operation (\autoref{defn:splitting-operation}), we have that $\w_{i,j,\LB}^{\alpha_{i,j}, \beta_{i,j}} = \w_{i,j}^{\alpha_{i,j}, \beta_{i,j}} = \w_{i,j,\TR}^{\alpha_{i,j}, \beta_{i,j}}$, i.e., the only possible common edge $\w_{i,j,\LB}^{\alpha_{i,j}, \beta_{i,j}}\to \w_{i,j,\TR}^{\alpha_{i,j}, \beta_{i,j}}$ between $R_i$ and $T_j$ is not an edge in $G_2$. Hence, $R_i$ and $T_j$ are edge-disjoint.
%\qed
\end{proof}

\subsection{Proof of~\autoref{thm:main-result}}
\label{subsec:proof-of-main-theorem}

Finally we are ready to prove our main theorem (\autoref{thm:main-result}) which is restated below:
\mainthm*
%\begin{reptheorem}{thm:main-result}
%\normalfont The \edp problem on planar DAGs is W[1]-hard parameterized by the number $k$ of terminal pairs. Moreover, under ETH, the \edp problem on planar DAGs cannot be solved $f(k)\cdot n^{o(k)}$  time where $f$ is any computable function, $n$ is the number of vertices and $k$ is the number of terminal pairs.
%\end{reptheorem}
%\repeattheorem{main}
\begin{proof}
Given an instance $(k,N,\mathcal{S})$ of \gtleq, we use the construction from~\autoref{subsec:construction} to build an instance $(G_2, \mathcal{T})$ of \edp such that $G_2$ is a planar DAG (\autoref{claim:G2-is-planar-and-dag}). It is easy to see that $n=|V(G_2)|=O(N^{2}k^{2})$ and $G_2$ can be constructed in $\poly(N,k)$ time.

%\todo{Add a paragraph about how to get in-degree and out-degree to be at most 2}

%\todo{Check if this is the same version as in the book, or the one with indices swapped}.
It is known
%\footnote{As noted in~\autoref{remark:issue-of-gtleq-swapping-indices}, our definition of \gtleq is slightly different than (but equivalent to)~\cite[Theorem 14.30]{fpt-book}}
~\citep[Theorem 14.30]{fpt-book} that \gtleq is W[1]-hard parameterized by $k$, and under ETH cannot be solved in $f(k)\cdot N^{o(k)}$ time for any computable function $f$. Combining the two directions from~\autoref{subsec:hard-direction} and~\autoref{subsec:easy-direction}, we get a parameterized reduction from \gtleq to an instance of \edp which is a planar DAG and has $|\mathcal{T}|=2k$ terminal pairs. Hence, it follows that \edp on planar DAGs is W[1]-hard parameterized by number $k$ of terminal pairs, and under ETH cannot be solved in $f(k)\cdot n^{o(k)}$ time for any computable function $f$.

Finally we show how to edit $G_2$, without affecting the correctness of the reduction, so that both the max out-degree and max in-degree are at most $2$. We present the argument for reducing the out-degree: the argument for reducing the in-degree is analogous. Note that the only vertices in $G_2$ with out-degree $>2$ are $A\cup C$. For each $c_j\in C$ we replace the directed star whose edges are from $c_j$ to each vertex of $\Le(G_{1,j})$ with a directed binary tree whose root is $c_i$, leaves are the set of vertices $\Le(G_{1,j})$ and each edge is directed away from the root. It is easy to see that in this directed binary tree the set of paths from $c_j$ to the different leaves (i.e.,vertices of $\Le(G_{1,j})$) are pairwise edge-disjoint, and we have only increased the number of vertices by $O(k)$ while maintaining both planarity and (directed) acyclicity. We do a similar transformation for each $a_i\in A$. It is easy to see that this editing adds $O(k^2)$ new vertices and takes $\poly(k)$ time, and therefore it is still true that $n=|V(G_2)|=O(N^{2}k^{2})$ and $G_2$ can be constructed in $\poly(N,k)$ time.
%\qed
\end{proof}

\section{Conclusion \& Open Questions}
\label{sec:conclusion}

In this paper we have shown that \edp on planar DAGs is W[1]-hard parameterized by $k$, and has no $f(k)\cdot n^{o(k)}$ algorithm under the Exponential Time Hypothesis (ETH) for any computable function $f$. The hardness holds even if both the maximum in-degree and maximum out-degree of the graph are at most $2$. Our result answers a question of Slivkins~\cite{DBLP:journals/siamdm/Slivkins10} regarding the parameterized complexity of \edp on planar DAGS, and a question of Cygan et al.~\cite{DBLP:conf/focs/CyganMPP13} and Schrijver~\cite{schrijver-building-bridges-II} regarding the parameterized complexity of \edp on planar directed graphs.

We now propose some open questions related to the complexity of the \disjp problem:
\begin{itemize}
  \item What is the \emph{correct} parameterized complexity of \edp on planar graphs parameterized by $k$? Can we design an XP algorithm, or is the problem NP-hard even for $k=O(1)$ like the general version? Note that to prove the latter result, one would need to have directed cycles involved in the reduction since there is $n^{O(k)}$ algorithm of Fortune et al.~\cite{DBLP:journals/tcs/FortuneHW80} for \edp on DAGs.
  \item Is the half-integral version\footnote{Each edge can belong to at most two of the paths} of \edp FPT on directed planar graphs or DAGs? It is easy to see that our W[1]-hardness reduction does not work for this problem.
  \item Given our W[1]-hardness result, can we obtain FPT (in)approximability results for the \edp problem on planar DAGs? To the best of our knowledge, there are no known (non-trivial) FPT (in)approximability results for any variants of the \disjp problem. This question might be worth considering even for those versions of the \disjp problem which are known to be FPT since the running times are astronomical (except maybe~\cite{DBLP:conf/stoc/LokshtanovMP0Z20}). Some of the recent work~\cite{julia-2,julia-3,julia-4,julia-5} on polynomial time (in)approximability of the \disjp problem might be relevant.
%  \item
\end{itemize}

%~\\

\subsection*{Acknowledgements}

We thank the anonymous reviewers of CIAC 2021 for their helpful comments. In particular, one of the reviewers suggested the strengthening of~\autoref{thm:main-result} for the case when the input graph has both in-degree and out-degree at most $2$.

%\newpage

%~\\
%
%\todo{Clean up references}

\bibliography{papers}
\bibliographystyle{plainnat}

\end{document}